\documentclass[pdflatex,sn-mathphys-num,referee]{sn-jnl}% Math and Physical Sciences Numbered Reference Style
\usepackage{academicons}
\usepackage{adjustbox}
\usepackage{algorithm}%
\usepackage{algorithmicx}%
\usepackage{algpseudocode}%
\usepackage{amsfonts}
\usepackage{amsmath,amssymb}
\usepackage{amsthm}
\usepackage[title]{appendix}%
\usepackage{bigints}
\usepackage{bm}
\usepackage{booktabs}
\usepackage{float}
\usepackage{graphicx}
\usepackage{geometry}
\usepackage{hyperref}
\usepackage{listings}%
\usepackage{mathtools}
\usepackage{multirow}%
\usepackage{setspace}
\usepackage{siunitx}
\usepackage[caption=false]{subfig}

\usepackage{textcomp}%
\usepackage{xcolor}

\usepackage{hyperref} %Must be added last
\graphicspath{{images/}}
\hypersetup{
    colorlinks=true,
    citecolor = blue,
    linkcolor=blue,
    filecolor=magenta,      
    urlcolor=cyan,
    pdfpagemode=FullScreen,
    }
    
\theoremstyle{thmstyleone}%
\newtheorem{theorem}{Theorem}[section]%  meant for continuous numbers
\newcommand{\balpha}{\bm{\alpha}}
\newcommand{\bbeta}{\bm{\beta}}

\newcommand{\balpham}{\bm{\alpha}_u}

\newcommand{\bD}{\mathcal{D}}
\newcommand{\bH}{\bm{H}}
\newcommand{\bpsi}{\bm{\eta}}

\newcommand{\bX}{\bm{X}}
\newcommand{\bY}{\bm{Y}}

\newcommand{\bxo}{\bm{x}_{o}}

\newcommand{\bxoi}{\bm{x}_{o,i}}
\newcommand{\bx}{\bm{x}}

\newcommand{\bz}{\bm{z}}
\newcommand{\CRA}{\mathbb{A}_p}

\newcommand{\ebg}{e^{\beta^{(s)}_{kg} }}

\newcommand{\ev}{\tilde{E}^{(r)}}
\newcommand{\evg}{\tilde{E}_{ikg}^{(r)}}
\newcommand{\ez}{\tilde{z}_{ig}^{(r)}}
\newcommand{\ismiss}{\mathbb{I}_{(k \in \mathcal{U} )} }
\newcommand{\isobs}{\mathbb{I}_{(k \in \mathcal{O} )} }

\newcommand{\mbmp}{\mathbb{B}_{p_{u}}}
\newcommand{\mbmpc}{\mathbb{B}^C_{p_{u}}}
\newcommand{\mbpi}{\mathbb{B}_{p_{u,i}}}
\newcommand{\mbmpci}{\mathbb{B}^C_{p_{u,i}}}
\newcommand{\ms}{\mathcal{U}}
\newcommand{\os}{\mathcal{O}}

\newcommand{\bTD}{\mathcal{TD}}
\newcommand{\tl}{\tilde{l}}

\theoremstyle{thmstyletwo}%
\numberwithin{theorem}{subsection}
\theoremstyle{thmstylethree}%

\definecolor{cerulean}{HTML}{2A52BE}
\begin{document}

\title[Article Title]{Handling Missingness and Censoring in Dirichlet Mixture Models}

\author*[1]{\fnm{Jason} \sur{Pillay}}\email{js.pillay@up.ac.za}
\author[1,4]{\fnm{Andri\"ette} \sur{Bekker}}\email{andriette.bekker@up.ac.za}
\author[2]{\fnm{Cristina} \sur{Tortora}}\email{cristina.tortora@sjsu.edu}
\author[3]{\fnm{Antonio} \sur{Punzo}}\email{antonio.punzo@unict.it }
\equalcont{These authors contributed equally to this work.}
%%==================================%%
%% abstract %%
%%==================================%%

\abstract{%What is the problem?
Incomplete compositional data analysis faces a fundamental limitation: likelihood-based methods for compositional models generally require fully observed compositions, making it difficult to accommodate missing or censored proportions directly on the simplex. Consequently, analysts often discard partially observed compositions or transform the data into unconstrained spaces, potentially sacrificing interpretability and coherence.
This paper proposes a likelihood-based method for incomplete compositional data without leaving the simplex. Specifically, we develop an Expectation–Maximisation (EM) type algorithm for fitting finite mixtures of Dirichlet distributions in the presence of missing and censored components. The proposed approach performs parameter estimation and model-based imputation simultaneously while preserving the compositional structure and interpretability of the original variables.

A simulation experiment evaluates the performance of the proposed estimators and imputations under increasingly complex coarsening mechanisms. Particular attention is paid to clustering performance, and model selection outcomes. The results showed beneficial clustering performance despite observations being incomplete, and a higher probability of model selection metrics identifying the correct number of clusters compared to current alternative of case-deletion.

The practical utility of the method is illustrated using two real datasets with distinct coarsened patterns. Analysis of the xenolith dataset identifies a four-component Dirichlet mixture that reveals interpretable profiles of rock types and speciation methods. Application to PM$_{2.5}$ speciation data from the Air Quality System, containing both left-censored and missing-at-random values, supports a four-component mixture model that characterises compositional parts of particulate matter across the United States.\\
%{\bf Highlights:} Some highlights in this paper include the following:
% \begin{itemize}
%     \item The estimation algorithm can reliably perform on multiple forms of coarsening, namely: coarsening completely at % random, missing at random, and censoring.
%     \item The algorithm is capable of fitting a finite mixture model on incomplete compositional data.
%     \item While transformation-based methods need at least two observations per row, the proposed algorithm only needs at % least one.
%     \item Imputations are made directly on the simplex instead of the real plane and back-transformed.
%     \item The Dirichlet distribution can be fitted on compositional datasets with no complete rows with the proposed % algorithm -- a scenario in which complete-case deletion would exclude the entire dataset.
% \end{itemize}
}

\maketitle


\section{Introduction}
\label{introduction}

Compositional data are multivariate observations describing the relative proportions of parts that constitute a whole. Because only the relative information carried by the components is informative, compositional vectors consist of strictly positive values constrained to sum to a fixed constant, typically one or one hundred. Consequently, compositional observations lie on the simplex rather than in unconstrained Euclidean space. This constant-sum constraint induces dependence among the components, so that increasing one component necessarily decreases the relative proportions available to the others. Conventional multivariate methods therefore cannot be applied directly without risking spurious correlations and misleading inference.

Compositional data arise naturally across numerous scientific disciplines. In geochemistry, the relative abundances of major oxides are routinely analysed to classify rocks, distinguish lithologies, and investigate geological processes. Mantle xenoliths, fragments of the Earth's mantle transported to the surface through volcanic activity, provide one of the few direct sources of information about the composition of the upper mantle \citep{mantle_minerology, xenolith_scale}. Their chemical compositions are therefore play an important role in studying mantle evolution, partial melting, and tectonic history. Likewise, atmospheric particulate matter consists of mixtures of chemical species whose relative abundances provide insight into emission sources, atmospheric processes, and the effectiveness of pollution-control strategies. Although these applications arise in different scientific fields, they share a common objective: identifying groups of observations with compositions. In this paper, these motivating examples are represented by a large mantle xenolith dataset obtained from the EarthChem repository and particulate matter speciation measurements collected through the United States Environmental Protection Agency's Air Quality System (AQS).

The identification of latent groups is frequently one of the primary objectives of compositional data analysis. In geochemistry, clusters may correspond to distinct mantle lithologies, petrogenetic processes, or tectonic environments. In atmospheric science, clusters may reflect different pollution sources or atmospheric conditions that produce characteristic chemical signatures. Finite mixture models provide a natural probabilistic approach for this task by representing the population as a collection of latent subpopulations, each characterised by its own compositional distribution. Besides producing probabilistic cluster assignments, finite mixture models estimate representative compositions for each group and permit likelihood-based selection of the number of clusters. Consequently, they have become an attractive approach for model-based clustering of heterogeneous compositional datasets.

Two broad strategies have been developed for modelling compositional observations. The most widely adopted approach transforms the data from the simplex to Euclidean space using additive log-ratio, centred log-ratio, isometric log-ratio, or related transformations before applying conventional multivariate statistical methods \citep{cell_proportions, microbiome, PM2_5}. These methods have been highly influential and have enabled the extension of numerous statistical techniques to compositional data. Nevertheless, scientific interpretation on Euclidean space is ultimately required on the original simplex, whereas estimation is performed in the transformed space. Furthermore, transformations generally require complete observations before they can be computed, making the treatment of missing values and censoring an additional preprocessing step \citep{logsticnormal_inflated, logisticnormal}.

An alternative is to model compositions directly on the simplex. Among simplex-supported probability distributions, the Dirichlet distribution occupies a prominent position because its support coincides exactly with the sample space of compositional data while remaining mathematically simple and computationally efficient. Its parameters admit direct interpretation through the expected composition and the concentration of observations around that expectation. These advantages have led to successful applications of the Dirichlet distribution and finite mixtures of Dirichlet distributions in fields including social media analysis, microbiology, behavioural science, and environmental studies \citep{dirichlet_nested_fmm, dirichlet_microbiome, dirichlet_fmm, dirichlet_zoid, dirichlet_pathology, dirichlet_emotion}. For clustering applications, modelling directly on the simplex also avoids the need to interpret cluster characteristics through transformed coordinates.

Despite these advantages, the practical application of Dirichlet mixture models is complicated by the widespread occurrence of incomplete observations. Incomplete observations can have two or more missing or censored values. Missing values arise through equipment malfunction, incomplete reporting, data integration, or sampling limitations, while censored observations commonly occur because analytical instruments cannot quantify concentrations below a certain threshold. Both motivating datasets considered in this paper exhibit these challenges. The EarthChem mantle xenolith dataset contains substantial proportions of missing measurements arising from the aggregation of geochemical data collected using different analytical techniques, whereas the AQS particulate matter data contain both missing observations and left-censored measurements resulting from instrument-specific detection limits.

Incomplete observations present a particular challenge for compositional data because the closure constraint links every component of the composition. Uncertainty regarding a single component therefore affects the interpretation of every remaining component. Existing approaches generally address this problem through preprocessing procedures such as deletion, substitution, multiplicative replacement, nearest-neighbour imputation, regression-based imputation, or multiple imputation before fitting the statistical model \citep{geo_mean, knn, svd, alr_em_method, review}. Although these methods are often effective in reconstructing plausible compositions, they separate the treatment of incomplete observations from the subsequent clustering analysis. Consequently, uncertainty introduced during imputation is typically not incorporated into parameter estimation or cluster allocation. Transformation-based approaches introduce an additional complication because imputation or replacement is usually required before the transformation itself can be applied.

This paper proposes an Expectation-Maximisation type algorithm for fitting a finite mixture of Dirichlet distributions to compositional data containing missing and censored observations. Parameter estimation, cluster allocation, and imputation are performed simultaneously within a single likelihood-based procedure operating directly on the simplex. The proposed method accommodates datasets containing no fully observed observations, requiring only a single observed component per composition, thereby extending the applicability of Dirichlet mixture models to highly incomplete datasets. Its performance is evaluated through an extensive simulation study and demonstrated using the EarthChem mantle xenolith and AQS particulate matter datasets, where the resulting clusters provide interpretable summaries of geochemical and environmental variation while naturally accounting for incomplete observations.

%\begin{itemize}
%\item[-] The proposed algorithm enables the fitting of \textcolor{red}{a Dirichlet distribution and} finite mixtures of Dirichlet distributions even when the dataset contains no fully observed rows, requiring only a single observed component per row.

%\item[-] An Expectation–Maximisation (EM) type algorithm is developed for fitting the Dirichlet distribution and its finite mixtures to compositional data with missing or censored entries. The complete-data and observed-data likelihoods are derived, and an efficient update procedure for the Dirichlet parameters is obtained, yielding a computationally attractive estimation framework well suited to high-dimensional settings.

%\item[-] The method performs parameter estimation and imputation jointly within a single likelihood-based framework, operating directly on the simplex without requiring transformations to the real space. This avoids the additional errors and interpretational limitations associated with back-transformation.

%\item[-] The approach accommodates both missing and censored observations within a unified modelling structure.

%\item[-] An extensive simulation study compares the proposed method with common competitors and evaluates parameter recovery, imputation accuracy, \textcolor{orange}{and model selection reliability }across varying levels of missingness.

%\item[-] An applications to a real dataset demonstrate improved interpretability and consistent inference on the original compositional variables.

%\end{itemize}

 The rest of the paper is structured as follows: Section \ref{prelim} reviews briefly the Dirichlet distribution and its theoretical properties, as well as the nature of missing values considered in this paper. These provide the necessary context to introduce the novelty of the paper in Section \ref{methodology}, namely the algorithm that fits a finite mixture Dirichlet model onto incomplete data. The algorithm's performance is assessed in Section \ref{simulations} in an extensive simulation experiment. The promising results from the simulation study motivates for a practical application of the algorithm in Section \ref{application}, followed by the conclusion in Section \ref{conclusion}.

\section{Preliminaries}
\label{prelim}

This section begins with a review of the Dirichlet distribution (Section \ref{dirichlet_review}), followed by a discussion of the various types of unobserved values (Section \ref{types_of_missing}). Throughout this paper, for a given dimension $p \in \mathbb{N}_{+}$, we define the following domains:
\begin{align*}
    \mathbb{R}_{+}^p 
    &= \left\{ \bm{\alpha} = [\alpha_1,\dots,\alpha_p]^{\top} 
    : \alpha_k > 0 \text{ for } k = 1,\dots,p \right\}, 
\end{align*}
the open unit simplex
\begin{align*}
    \mathbb{V}_p 
    &= \left\{ \bx = [x_1,\dots,x_p]^{\top} 
    : \text{$x_k \ge 0$ for $k = 1,\dots,p$ with $\sum_{k=1}^p x_k < 1$} \right\},
\end{align*}
and the standard simplex
\begin{align*}
    \mathbb{S}_p 
    &= \left\{ \bx = [x_1,\dots,x_p]^{\top} 
    : \text{$x_k \ge 0$ for $k = 1,\dots,p$ with $\sum_{k=1}^p x_k = 1$}
     \right\}.
\end{align*}

\subsection{The Dirichlet distribution: a brief review}
 \label{dirichlet_review}
    A random vector $\bX\in\mathbb{S}_p$ is said to follow a Dirichlet distribution with parameter vector $\balpha = [\alpha_1,\ldots,\alpha_p]^{\top} \in \mathbb{R}_+^p$, denoted by $\bX \sim \bD_{\mathbb{S}_p}(\balpha)$, if its probability   density function (density) is
    \begin{equation}
    \label{pdf_dirichlet}
        f_{\bD_{\mathbb{S}_p}}(\bx;\balpha) 
        =
           \dfrac{\Gamma(\alpha_0)}
            {\displaystyle\prod_{k=1}^p \Gamma(\alpha_k)}
            \displaystyle\prod_{k=1}^p x_k^{\alpha_k-1},\qquad 
            \bx \in \mathbb{S}_p, 
    \end{equation}
    where $\alpha_0 = \|\balpha\|_1 = \displaystyle\sum_{k=1}^p \alpha_k$, and $\Gamma(\cdot)$ denotes the gamma function.
Since $\displaystyle\sum_{k=1}^p X_k = 1$, the last component of $\bX$---though, in principle, any component could be chosen--can be expressed as
    $$
    X_p = 1 - \sum_{k=1}^{p-1} X_k.
    $$
Hence, $\bX \sim \bD_{\mathbb{S}_p}(\balpha)$ can equivalently be represented by its first $p-1$ components, say $\bX_{-p} =    [X_1,\ldots,X_{p-1}]^{\top}\in\mathbb{V}_{p-1}$, whose density is
    \begin{equation}
        f_{\bD_{\mathbb{V}_{p-1}}}(\bx_{-p};\balpha) 
        =
            \dfrac{\Gamma(\alpha_0)}
            {\displaystyle\prod_{k=1}^p \Gamma(\alpha_k)}
            \left(1 - \|\bx_{-p}\|_1 \right)^{\alpha_p-1}\displaystyle\prod_{k=1}^{p-1} x_k^{\alpha_k-1}
            , \qquad
            \bx_{-p} \in \mathbb{V}_{p-1}. 
            \label{eqref:Dirichlet2}
    \end{equation}
    In this case, we will write $\bX_{-p} \sim \bD_{\mathbb{V}_{p-1}}(\balpha)$.
    Note that, although \eqref{pdf_dirichlet} and \eqref{eqref:Dirichlet2} represent
equivalent formulations of the same density, they are defined on different domains, namely $\mathbb{S}_p$ and $\mathbb{V}_{p-   1}$,
    respectively. 
    See \citet[Chapter~2]{dirichlet} for further details on the Dirichlet distribution.

\subsection{The truncated Dirichlet distribution}
\label{truncated_dirichlet_review}

Let $\CRA \subset \mathbb{S}_p$ be a measurable set with
$\Pr\left(\bY \in \CRA\right) > 0$, where $\bY \sim \bD_{\mathbb{S}_p} (\balpha)$.
A random vector $\bX$ is said to follow a truncated Dirichlet distribution on $\CRA$, with parameter vector $\balpha \in \mathbb{R}_+^p$,
denoted $\bX \sim \bTD_{\CRA}(\balpha)$, if
\[
\bX \stackrel{d}{=} \bY \mid \{\bY \in \CRA\}.
\]
The density of $\bX \sim \bTD_{\CRA}(\balpha)$ is
\begin{align}
\label{truncated_pdf}
f_{\bTD_{\CRA}}(\bx;\balpha) 
=
\frac{f_{\bD_{\mathbb{S}_p}}(\bx;\balpha)}{F(\CRA; \balpha)},
%\, \mathbb{I}_{\CRA}(\bx),
%\quad \bx \in \mathbb{S}_p,
\qquad \bx \in \CRA,
\end{align}
where
\[
F(\CRA; \balpha)
=
\Pr(\bY \in \CRA), 
\qquad \bY \sim \bD_{\mathbb{S}_p}(\balpha),
\]
is the Dirichlet probability of the truncation region $\CRA$.
Throughout this paper, $\CRA$ is assumed to be defined by lower and/or upper bounds on the components of the random vector. Specifically,
\begin{align}
\CRA 
&= \left\{ \bx \in \mathbb{S}_p 
: a_{kL} \le x_k \le a_{kU}, \; k = 1,\dots,p \right\} \nonumber\\
&= \left\{ \bx \in \mathbb{S}_p 
: \bm{a}_L \le \bx \le \bm{a}_U \right\}, \label{eq:Ap}
\end{align}
where $\bm{a}_L = [a_{1L},\ldots,a_{pL}]^\top$ and 
$\bm{a}_U = [a_{1U},\ldots,a_{pU}]^\top$ denote the vectors of lower and upper bounds, respectively. 
The bounds satisfy $0 \le a_{kL} < a_{kU} \le 1$ for $k=1,\ldots,p$, and the inequalities are understood componentwise.
See \citet[][Chapter~7]{dirichlet} for further details on the truncated Dirichlet distribution.

However, in compositional data analysis it is common to encounter observations with partially unobserved components. 
Accounting for such components is therefore essential when modelling Dirichlet-distributed data. 
Consequently, the model in \eqref{pdf_dirichlet} must be adapted to accommodate incomplete observations. 
To this end, we partition the random vector $\bX$ into its unobserved and observed components,
\[
\bX =
\begin{bmatrix}
\bX_u \\
\bX_o
\end{bmatrix},
\]
where the subscripts $m$ and $o$ denote the missing and observed parts, respectively.
Let $\mathcal{U} \subseteq \{1,\dots,p\}$ denote the index set of unobserved components, with cardinality $p_u = |\mathcal{U}|$, and let $\mathcal{O} = \{1,\dots,p\} \setminus \mathcal{U}$ denote the index set of observed components, with cardinality $p_o = |\mathcal{O}|$. 
Accordingly, $\bX_u$ and $\bX_o$ denote the subvectors of $\bX$ formed by the components indexed by $\mathcal{U}$ and $\mathcal{O}$, respectively.
The missing component $\bX_u$ is assumed to lie in the region
\begin{equation}
\label{eq:Bpu}
\mathbb{B}_u
=
\left\{
\bX_u \in \mathbb{R}_+^{p_u}
:
b_{kL} \le x_k \le b_{kU},
\; k \in \mathcal{U}
\right\},    
\end{equation}

where $b_{kL}$ and $b_{kU}$ denote lower and upper bounds for the $k$-th missing component.
Under this formulation, the conditional distribution of the unobserved component given the observed component as well as the distribution of the observed part (i.e., the marginal) both admit closed forms. Theorem~\ref{dist} presents the resulting distributions for $\bX_u \mid \bX_o$ and $\bX_o$.
%%%%%%%%%%%%%%%%%%
\begin{theorem}
\label{dist}
Let $\bX \sim \bD_{\mathbb{S}_p}(\balpha)$ and partition $\bX$ and $\balpha$ according to the index sets $\mathcal{U}$ and $\mathcal{O}$ as
\[
\bX =
\begin{bmatrix}
\bX_u \\
\bX_o
\end{bmatrix},
\qquad
\balpha =
\begin{bmatrix}
\balpham \\
\balpha_o
\end{bmatrix},
\]
where $p_u = |\mathcal{U}|$ and $p_o = |\mathcal{O}|$.
Define
\[
\balpha_o^*
=
\begin{bmatrix}
\balpha_o \\
\|\balpha\|_1 - \|\balpha_o\|_1
\end{bmatrix}.
\]
Then the marginal density of the observed component $\bX_o$ is given by
\begin{equation}
\label{marginal}
f(\bxo;\balpha)
=
f_{ {\bD_{\mathbb{V}_{p_o-1}}} }(\bxo;\balpha_o^*)\;
F_{\bD_{\mathbb{S}_{p_u}}}\left(\mbmp;\balpha_u \mid \bxo\right),
\qquad \bxo \in \mathbb{V}_{p_o}, 
\end{equation}
where $\mbmp$ is defined as in \eqref{eq:Bpu} and
\begin{equation*}
F_{\bD_{\mathbb{S}_{p_u}}}\left(\mbmp;\balpha_u \mid \bxo\right)
=
\Pr\left(\bX_u \in \mbmp \,\middle|\, \bX_o=\bxo\right).
% &=
% \Pr\left(\frac{\bX_u}{c(\bxo)} \in \mbmpc \,\middle|\, \blue{\bX_o=\bxo}\right).
\end{equation*}
Moreover, letting $c(\bxo)=1-\|\bxo\|_1$, the conditional distribution of the unobserved component satisfies
\begin{equation}
\label{cond}
\frac{\bX_u}{c(\bxo)}
\;\bigg|\;
\bX_o=\bxo
\sim
\bTD(\balpham)
\quad
\text{on } \mbmpc,
\end{equation}
where $\mbmpc$ denotes the region $\mbmp$ rescaled by the closure factor $c(\bxo)$, namely
\[
\mbmpc
=
\left\{
\bz \in \mathbb{S}_{p_u}
:
\frac{b_{kL}}{c(\bxo)} \le z_k \le \frac{b_{kU}}{c(\bxo)},
\; k \in \mathcal{U}
\right\}.
\]
\end{theorem}
%%%%%%%%%%%%%
\begin{proof}
The result follows from the marginalisation property of the Dirichlet distribution and its conditional Dirichlet structure under compositional constraints; see archive \cite{jasa} for the proof.
\end{proof}

\subsection{Missingness, censoring, and coarsening at random}
\label{types_of_missing}

Understanding how incomplete data are handled in this paper requires distinguishing between two related but different concepts: the mechanism, and the form in which the data are observed. From the motivation around the perseverance of unobserved data discussed in Section \ref{introduction}, compositional data may be unobservable due to a part being too small to measure with current techniques and machinery, or too large that it dominates the rest of the composition. In other instances, parts may be unobserved due to non-response. The former thus, is more informed on the region in which the unobserved value lies compared to the latter. That is, the entire composition is not fully observed, but also not perfectly missing. Thus, the unobserved part falls part of a more general notion, referred to as coarsened data. Further, the setting in which datasets are subject to coarsening is typically a consequence of deterministic thresholds. This implies that any value outside the threshold regions are hidden, regardless of how close or how far said value is to the threshold value. Thus, the coarsening is said to be done at random. Under this assumption, we are afforded two special cases of coarsening: (1) is the typical deterministic censoring as discussed, and (2) a value that is missing at random (MAR). Notice that a value that is missing at random is itself a special, albeit trivial, type of censoring where the threshold values are extended to the endpoints of possible values the random variable can assume or beyond. Thus MAR mechanism (and all its special cases) is contained within the CAR mechanism. Furthermore, the coarsening at random assumption encompasses a special case of coarsening completely at random (CCAR), which similarly to CAR, contains the MCAR missingness as a special case.  The framework thus incorporates various mechanisms behind unobserved values while relying on the standard ignorability conditions required for likelihood-based inference. For further descriptions, see \cite{jasa}.

\section{Parameter estimation}\label{methodology}
When the observed proportions arise from multiple underlying populations, a finite mixture of Dirichlet distributions provides a flexible modelling framework. 
The mixture density is given by
\begin{align}
    f_{M\bD}(\bx;\bpsi) = \sum_{g=1}^G \pi_g f_{\bD_{\mathbb{S}_p}},
    \label{fmm}
\end{align}
where $\bpsi = \{\pi_g,\balpha_g : g=1,\dots,G\}$. 
The mixing proportions satisfy $\pi_g>0$ and $\sum_{g=1}^G \pi_g = 1$, and $f_{\bD_{\mathbb{S}_p}}(\cdot)$ is defined in \eqref{pdf_dirichlet}.

\subsection{EM algorithm: Finite mixtures of Dirichlet distributions}
\label{inference}
A random sample $\mathcal{X}$ that is generated by \eqref{fmm} is incomplete for two reasons: (1) the row potentially lack some observations, and (2) it is not known which of the $G$ components the observation is generated from. There is a latent cluster membership vector $\bz_i = \begin{bmatrix} z_{i1},\dots, z_{iG}\end{bmatrix}^{\top}$ where $z_{ig} = 1$ if $\bx_i$ is generated from the $g^{th}$ component of model \eqref{fmm} and 0 otherwise, subjected to the constraint $||\bz_i||_1 =1 $, for $i=1,\dots,n$. 
Letting $\mathcal{Z} = \{\bz_i: i =1,\dots,n \}$, the complete likelihood is therefore given as
\begin{align}
    \tilde{\mathcal{L}}(\bpsi; \mathcal{X}, \mathcal{Z} ) 
    &= \prod_{i=1}^n \prod_{g=1}^G  \left\{ \pi_g f_{\bD_{\mathbb{S}_p}}(\bx_i;\balpha_g) \right\}^{z_{ig}} & \nonumber\\
    &= \prod_{i=1}^n \prod_{g=1}^G \left\{ \Gamma\left(  \alpha_{0g} \right) \left[\displaystyle\prod_{k=1}^p \Gamma(\alpha_{kg}) \right]^{-1}\displaystyle\prod_{k \in \ms} x_k^{\alpha_{kg} -1} \displaystyle\prod_{k \in \os} x_k^{\alpha_{kg} -1} \right\}^{z_{ig}}. \label{likelihood_fmm}
\end{align}

The complete log-likelihood is then the natural logarithm of \eqref{likelihood_fmm} and is given as
\begin{align}
    \tl_c(\bpsi;\mathcal{X},\mathcal{Z}) & = \sum_{i=1}^n \sum_{g=1}^G z_{ig}\pi_g + \sum_{i=1}^n \sum_{g=1}^G z_{ig}\ln\Gamma\left( \alpha_{0g} \right)  - \sum_{i=1}^n \sum_{g=1}^G\sum_{k=1}^p z_{ig}\ln\Gamma\left( \alpha_{kg}\right) +\sum_{i=1}^n \sum_{g=1}^G\sum_{k \in \ms}\alpha_k z_{ig}\ln x_{ik}  \nonumber\\
    & + \sum_{i=1}^n \sum_{g=1}^G\sum_{k \in \os}\alpha_k z_{ig} \ln x_{ik} -\sum_{i=1}^n \sum_{g=1}^G\sum_{k \in \ms}z_{ig}\ln x_{ik} - \sum_{i=1}^n \sum_{g=1}^G\sum_{k \in \os} z_{ig}\ln x_{ik}.
    \label{llfmm}
\end{align}
The E-step computes $\tilde{Q}(\bpsi ) = \mathbb{E}\left[\tl_c(\bpsi|\mathcal{X},\mathcal{Z})\big|\bpsi^{(r)}, \mathcal{X}_o\right]$ using the parameter updates at the $r^{th}$ iteration, namely $\bm{\bpsi}^{(r)}$. Taking the expected of \eqref{llfmm}, $\tilde{Q}(\bpsi )$ yields:
\begin{align}
\label{Q function fmm}
    \tilde{Q}(\bpsi ) & = \sum_{i=1}^n \sum_{g=1}^G \tilde{z}^{(r)}_{ig}\pi_g + \sum_{i=1}^n \sum_{g=1}^G \tilde{z}^{(r)}_{ig}\ln\Gamma\left( \alpha_{0g} \right)  - \sum_{i=1}^n \sum_{g=1}^G\sum_{k=1}^p \tilde{z}^{(r)}_{ig}\ln\Gamma\left( \alpha_{kg}\right) +\sum_{i=1}^n \sum_{g=1}^G\sum_{k \in \ms}\alpha_{kg} \tilde{z}^{(r)}_{ig}\evg  \nonumber\\
    & + \sum_{i=1}^n \sum_{g=1}^G\sum_{k \in \os}\alpha_{kg} \tilde{z}^{(r)}_{ig} \ln x_{ik} -\sum_{i=1}^n \sum_{g=1}^G\sum_{k \in \ms}\tilde{z}^{(r)}_{ig}\evg - \sum_{i=1}^n \sum_{g=1}^G\sum_{k \in \os} \tilde{z}_{ig}^{(r)}\ln x_{ik},
\end{align}
with the expected values of the two unobserved quantities, for $i=1,\dots,n$ and $g=1,\dots,G$: $ \tilde{z}_{ig}^{(r)} = \mathbb{E}[Z_{ig}| \bxoi, \bpsi^{(r)}, \mbpi]$ and $\evg = \mathbb{E}[\ln X_{u,ik}|z_{ig}, \bxoi, \bpsi^{(r)}, \mbpi]$. 

As in the classical use of the EM algorithm for finite mixture models, we have:
\begin{align}
\label{ez}
    \tilde{z}_{ig}^{(r)} = \mathbb{P}\left(z_{ig}=1|\bxoi, \bpsi^{(r)} \right) = \frac{ \pi^{(r)}_g f_{\bD_{\mathbb{S}_{p_o}}}\left(\bxoi; \balpha_g^{(r)} \right)F\left( \mbpi ; \balpha_{u,g}| \bxoi\right)}{ \displaystyle \sum_{g=1}^G \pi^{(r)}_g f_{\bD_{\mathbb{S}_{p_o}}}\left(\bxoi; \balpha_g^{(r)} \right)F\left( \mbpi ; \balpha_{u,g}| \bxoi\right)},
\end{align}
where $f_{\mathcal{D}}(\cdot)$ is given in \eqref{pdf_dirichlet}, $F\left( \mbpi ; \balpha_{u,g}| \bxoi\right)$ is the probability that $\bX_{m,i}$ falls within the censored region, $\mbpi$, conditioned on the observed component, $\bxoi$ and $\balpha_{u,g}^{(r)}$ is a vector of parameters from $\balpha_g^{(r)}$ corresponding to the unobsserved entries of $\bx_i$. In other words, the expected value given by \eqref{ez} is the posterior probability that $\bx_i$ belongs to the $g^{th}$ cluster at the $r^{th}$ iteration. Notice that if $z_{ig} =1$ it is then known that $\bx_i$ is generated from the $g^{th}$ component of model \eqref{fmm}. Therefore, the distribution of the unobserved part of $\bx_i$ has the following distribution, when conditioned on the observed part, $\bxoi$ and the cluster membership $z_{ig}$:
\begin{align*}
    \frac{\bX_{u,i}}{1 - ||\bxoi||_1} | \bxoi, z_{ig}=1  \sim \bTD_{\mbmpc}(\balpha_{u,g}),
\end{align*} 
for $i = 1,\dots,n$ and $g=1,\dots,G$ (refer to \eqref{truncated_pdf}). 
The derivation of $\evg$ now follows. From Theorem \ref{dist}, we have that 
\begin{align*}
    \evg =  \ln \left( 1 - ||\bxoi||_1 \right) + \frac{\partial {\ln} F\left(\mbmpci;\balpha_{u,g}^{(r)}\right)}{\partial\alpha_{kg}}  + \psi\left(\alpha_{kg}^{(r)}\right) - \psi\left(\left|\left|\balpha_{u,g}^{(r)}\right|\right|_1\right),
\end{align*} 
for $i = 1,\dots,n,~g=1,\dots,G,$ and $k \in \ms$. The M-step at the $(r+1)^{th}$ iteration maximises \eqref{Q function fmm} with respect to $\pi_g$ and $\balpha_g$. Derivatives with respect to the former does not involve the latter and vice versa, so \eqref{Q function fmm} can be maximised with respect to $\pi_g$ separately from $\balpha_g$. Maximising $\tilde{Q}$ with respect to $\pi_g$, subjected to the constraint mentioned in \eqref{fmm} updates $\pi_g^{(r+1)}$ as
\begin{align}
    \pi_g^{(r+1)} = \frac{ n_g }{n},
\end{align}
where $ n_g = \displaystyle\sum_{i=1}^n \ez $, for $g=1,\dots,G$. We transform the parameters for the $g^{th}$ cluster as follows, for $g=1,\dots,G$:
\begin{align}
    e^{\beta_{kg}} = \alpha_{kg}  \iff \ln \alpha_{kg} =  \beta_{kg},
\end{align}
where $\beta_{kg}$ is now an unconstrained parameter for $k = 1,\dots,p$ and $g=1,\dots,G$. \label{sup_mat}
The $k^{th}$ element of the unconstrained gradient, for the $g^{th}$ component of \eqref{fmm} is:
\begin{align}
\label{grad_beta_fmm}
      \nabla_{kg} = \frac{\partial \tilde{Q}(\bm{\eta})}{\partial\beta_{kg}} =  n_ge^{\beta_{kg}}\psi\left(\sum_{k=1}^p e^{\beta_{kg}} \right) - n_g e^{\beta_{kg}}\psi\left( e^{\beta_{kg}} \right) + e^{\beta_{kg}} \sum_{i=1}^n \ismiss \ev_{ik} \tilde{z}_{ig}^{(r)} + e^{\beta_{kg}} \sum_{i=1}^n \isobs\tilde{z}_{ig}^{(r)}\ln x_{ik} ,
\end{align}
for $g=1,\dots,G$. The second and mixed derivatives are:
 \begin{align*}
 \frac{\partial \tilde{Q}(\bm{\eta})}{\partial\beta_{jg}\partial \beta_{kg}} = 
     \begin{cases}
         n_g e^{\beta_{jg} + \beta_{kg}}\psi'\left(\displaystyle\sum_{k=1}^pe^{\beta_{kg}} \right) & \text{ if } j \ne k \\
          n_g e^{2\beta_{kg}}\psi'\left(\displaystyle\sum_{k=1}^pe^{\beta_{kg}} \right) -n_g e^{2\beta_{kg}} \psi'\left(e^{\beta_{kg}}\right) + \nabla_{kg}& \text{ if } j = k,
     \end{cases}
 \end{align*}
 for $g=1,\dots,G$. Allowing for the following notation at the $s^{th}$ iteration:

 \begin{align*}
     \gamma_g^{(s)}      & = n_g \psi'\left(\displaystyle\sum_{k=1}^p \ebg \right),\nonumber \\
     \nabla_{kg}^{(s)}  & = n_g\ebg\psi\left(\sum_{k=1}^p \ebg \right) - n_g\ebg \psi\left(\ebg \right)  + \ebg\sum_{i=1}^n \ismiss\ez \ev_{ik} + \ebg\sum_{i=1}^n \isobs \ez \ln x_{ik},\nonumber\\
     \bm{\beta}_{eg}^{(s)}   & =  \begin{bmatrix} e^{\beta_{1g}^{(s)}}, \dots, e^{\beta_{pg}^{(s)}} \end{bmatrix}^{\top},  \text{ and}&\\
     \bm{D}_g^{(s)} & = diag \left[ \nabla_{1g}^{(s)} - n_g e^{2\beta_{1g}^{(s)}} \psi'\left(e^{\beta_{1g}^{(s)}}\right), \dots,  \nabla_{pg}^{(s)} - n_g e^{2\beta_{pg}^{(s)}} \psi'\left(e^{\beta_{pg}^{(s)}}\right) \right],&
 \end{align*}
In the NR algorithm, $\bm{\beta}_g$ is updated via the NR algorithm through the following formula for $g=1,\dots,G$:
\begin{align}
\label{nr_start_fmm}
    \bm{\beta}_g^{(s+1)} = \bm{\beta}_g^{(s)} - (\bH_g^{(s)})^{-1}\bm{\nabla}_g^{(s)},
\end{align}
where 
 \begin{align*}
     \bH_g^{(s)} = \bm{D}_g^{(s)} + \gamma^{(s)}_g \bm{\beta}_{eg}^{(s)} \bm{\beta}_{eg}^{(s) \top}.
 \end{align*}
 It's corresponding element-wise update is derived as:
\begin{align}
    \beta_{kg}^{(s+1)} = \beta_{kg}^{(s)} - \frac{ \nabla^{(s)}_{kg} -w_g^{(s)} e^{\beta_{kg}^{(s)}} }{ d_{kkg}^{(s)}  },
\end{align}
where 
\begin{align}
      w_g^{(s)}    & = \frac{ \displaystyle \sum_{k=1}^p \frac{\ebg \nabla_{kg}^{(s)} }{d^{(s)}_{kkg}  } }  {\frac{1}{\gamma_g^{(s)} } +\displaystyle \sum_{k=1}^p \frac{e^{2\beta_k^{(s)}   } }{ d^{(s)}_{kk}} }. \nonumber
\end{align}
At convergence of the NR algorithm, $\alpha_{kg}$ is updated as:
\begin{align*}
    \alpha^{(r+1)}_{kg} = \ebg.
\end{align*}

\subsection{Initialisation and convergence}
\label{initialisation}

The Newton--Raphson (NR) algorithm requires suitable starting values in order to maximise the chances of convergence to the global maximum of the log-likelihood. 
When the NR algorithm is applied directly to the Dirichlet parameters $\balpha$, 
it has been noted that setting each $\alpha_k$ equal to the minimum value of the 
$k^{\text{th}}$ data column prevents the algorithm from producing negative updates 
when the method-of-moments estimates are close to zero \citep{initialisation_roning}.
In contrast, under the $\bbeta$-reparameterisation used in the M-step, this difficulty does not arise because the optimisation is carried out in a different parameter space.As a result, the method-of-moments estimates provide simple and reliable starting values for the NR algorithm \citep{estimation_NR}. This moment-based initialisation requires only the first two raw moments of each component of the Dirichlet distribution. Let $\mathcal{O}_k = \{i : x_{ik} \text{ is observed}\}$ denote the set of observations for which the $k^{\text{th}}$ component is available, and let $n_k = |\mathcal{O}_k|$. 
The $j^{\text{th}}$ raw sample moment of the $k^{\text{th}}$ component is computed using the available observations as
\begin{align}
\overline{x}_{jk}
=
\frac{1}{n_k}
\sum_{i \in \mathcal{O}_k} x_{ik}^{\,j},
\qquad j=1,2 .
\end{align}
The method-of-moments estimate of the $k^{\text{th}}$ element of $\balpha$ is then given by \cite{dirichlet}
\begin{align}
\overline{\alpha}_k =
\frac{\overline{x}_{11}-\overline{x}_{21}}
{\overline{x}_{21}-\overline{x}_{11}^2}
\,\overline{x}_{1k},
\qquad k=1,\dots,p .
\end{align}

Two stopping rules must be specified for the proposed algorithm. 
An inner stopping rule controls the convergence of the Newton--Raphson (NR) procedure used within the M-step, while an outer stopping rule determines the convergence of the overall EM algorithm.

For the NR algorithm, convergence is monitored using the quantity $\frac{1}{2}\lambda_g^2$  where:
\[ \lambda_g  = \sqrt{ \left(\Delta \bm \beta_g^{(s+1)}  \right)^{\top} \bm H^{(s+1)} \Delta \bm\beta_g^{(s+1)}    },
\]
with $\Delta \bm \beta_g^{(s+1)} = \bm \beta_g ^{(s+1)} - \bm \beta_g ^{(s)}$ and $\bm H^{(s+1)}$ is the Hessian matrix evaluated at $\bm \beta_g ^{(s+1)}$. The criterion, $\frac{1}{2}\lambda_g^2$ provides an approximation to the distance from the current objective function value to its supremum, namely $\left| Q(\bm{\beta}_g) - \underset{\bm\beta}{\mathrm{sup}}Q(\bm\beta)\right|$. Thus, iterations are stopped once $\frac{1}{2}\lambda_g^2$  falls below a small tolerance when iterating over the estimator for the $g^{th}$ cluster.
The observed-data log-likelihood increases monotonically under the EM algorithm. However, temporary stability may occur when the algorithm approaches a local maximum before moving towards the global maximum. To assess convergence to the asymptotic log-likelihood, we employ the Aitken acceleration criterion. Let $l_o^{(r)}$ denote the observed log-likelihood at the $r^{\text{th}}$ EM iteration. The Aitken acceleration is defined as
\begin{align*}
a^{(r+1)} =
\frac{l_o^{(r+2)}-l_o^{(r+1)}}{l_o^{(r+1)}-l_o^{(r)}} .
\end{align*}
The corresponding estimate of the asymptotic log-likelihood is
\begin{align*}
(l_o^{\infty})^{(r)} =
l_o^{(r+1)} +
\frac{l_o^{(r+2)}-l_o^{(r+1)}}{1-a^{(r+1)}} .
\end{align*}
The EM algorithm is considered to have converged when
\[
(l_o^{\infty})^{(r)}-l_o^{(r+1)} < \epsilon,
\]
where $\epsilon>0$ is a small tolerance \citep{convergence_aitken}.

\section{Simulation experiment}
\label{simulations}
The goal of the simulation experiment is to evaluate two aspects of the proposed algorithm: (1) the clustering capability of incomplete compositions, and (2) the model selection performance. From the E-step in Section \ref{inference}, the proposed algorithm is capable of fitting a mixture of Dirichlet distributions on data with increasingly more complex levels of incompleteness ranging from CCAR to censored entries. Accordingly, the simulation design examines the performance of the algorithm’s estimators when the data contain CCAR values and censored observations. In the experiment, datasets, each of size $n=100$, are generated as a mixture of $G=4$ Dirichlet distributions. The choice of clusters is to reflect the results obtained in the application in Section \ref{application}. Similarly, the parameter choice is computed according to the degree of pairwise cluster overlap apparent in the data applications. The dimension of the data's rows are varied from $p=6$ to $p=10$ to capture the clustering and model selection capabilities between the two datasets. The ability of the algorithm to accurately cluster incomplete observations using the model given in \eqref{fmm} is investigated, and the experiment examines the algorithm's ability to recover the true number of clusters present in the dataset via different model selection criteria. Thus, the experiment is divided into two parts, namely A and B:\\
\begin{itemize}
    \item In Part A, 1000 datasets are simulated, in which a percentage of entries are hidden completely at random for each dataset. The proposed algorithm is then fitted and evaluated on how accurately it could cluster the incomplete observations. Part A concludes with reporting the proportion of datasets the model selection metrics could correctly identify 4 components, and this is compared to that of mixtures fitted on the data's subset of complete observations.
    \item In Part B, clustering performance and model selection is assessed as in Part A, but 100 datasets simulated and subjected to censoring at varying quantiles.
\end{itemize}
Measuring clustering performance relies on how well the model is at classifying an observation to its correct cluster. Thus, the  simulation computes the clustering accuracy for each scenario. Of course, the accuracy does not take into account errors due to label-switching and classifying a point correctly by chance. Thus, the Adjusted Rand Index (ARI) is also computed to address these possible confounders. The ARI has a value equal to 0 when classification is no better than can be expected by random agreement between two partitions, and a value equal to 1 when there is a perfect agreement \citep{ari_review, ari_source}. The accuracy score here is defined as the proportion of rows correctly classified. The ARI's expression can be found in \citep{ari}, and the \texttt{ARI} function in the \texttt{MixGHD} package is used in this study. At convergence, we can establish whether an observation belongs to a component via the maximum a posteriori (MAP) probabilities. Letting $\hat{z}_{ig}$ denote $\ez$ at convergence, the predicted cluster membership of observation $\bx_i$ is determined as:
\begin{align}
\label{responsibilities}
    \hat{z}_i =\underset{g}{ \mathrm{argmax}}~\hat{z}_{ig},
\end{align}
In practice, the true number of components are unknown and would have to be inferred via model selection metrics. Thus, this simulation determines how often model selection criteria would be able to successfully determine the correct number of components present in the dataset. For each scenario, the proportion of datasets that were successfully identified as having four components by the: Akaike Information Criterion (AIC), Bayesian Information Criterion (BIC), and Integrated Classification Likelihood (ICL) are computed. The metrics are defined in \eqref{gof_metrics}, where $K$ denotes the number of free parameters:

\begin{align}
    \label{gof_metrics}
    AIC = 2l_o -2K \hspace{0.7cm} BIC = 2l_o -K\ln n \hspace{0.7cm} ICL = BIC + 2\displaystyle\sum_{i=g}^G \sum_{i=1}^n \hat{z}_{ig} \ln \hat{z}_{ig}  
\end{align}

\subsection{Simulation results: Part A}
For each dataset generated, values for the first $p-1$ columns are hidden according to the percentages 0\%, 10\%, 20\%, $\dots$, and 90\% for $p=6,\dots 10$. For each percentage of values hidden, the average of the resulting 1000 metrics are reported. 

\begin{figure}[H]
    \centering
    \includegraphics[ width=0.8\linewidth]{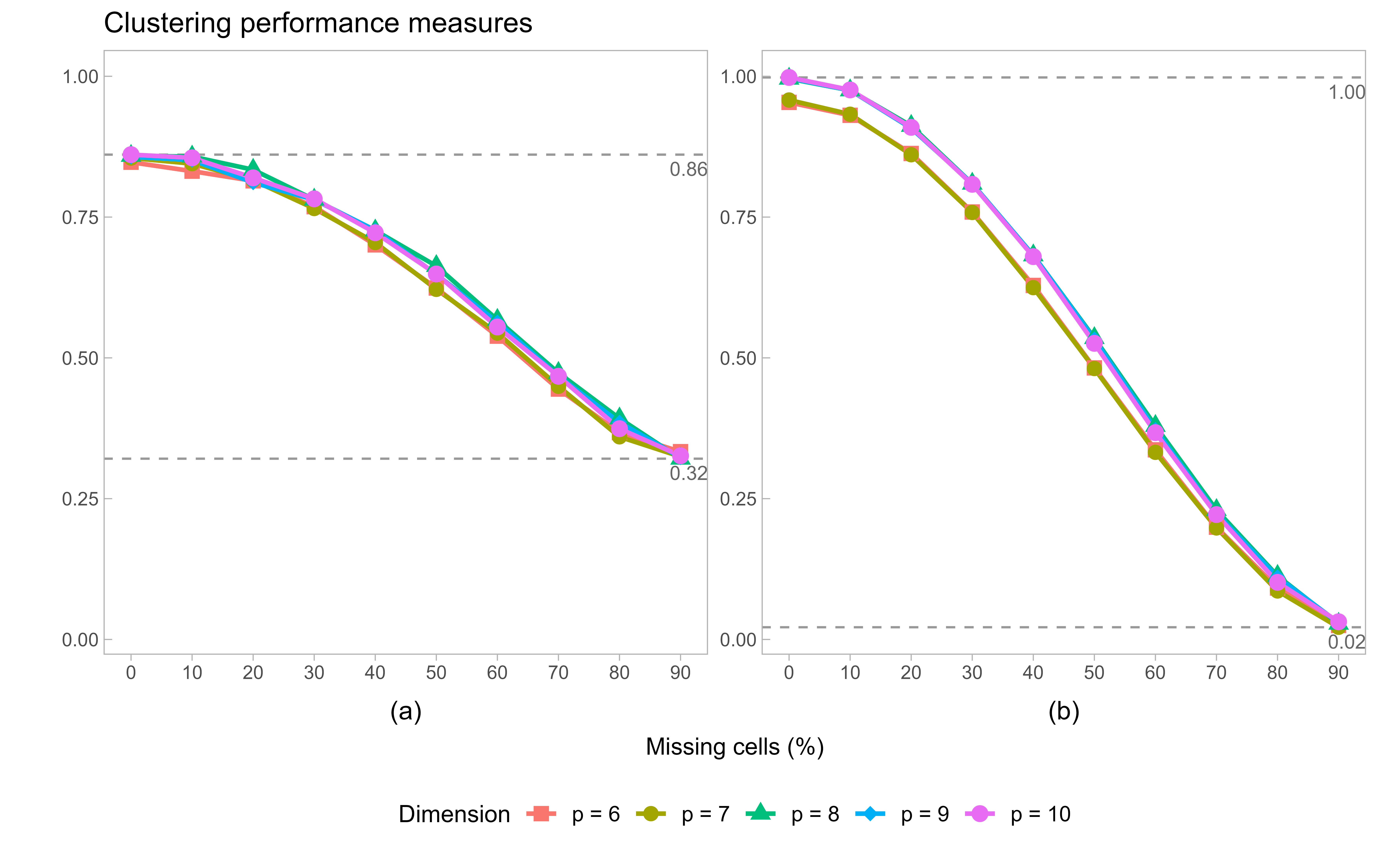}
    \caption{Cluster performance metrics: (a) Average accuracy scores and (b) average ARI scores for simulation Part A.}
    \label{acc_ari_plot}
\end{figure}

\begin{figure}[H]
    \centering
    \includegraphics[trim={0cm 2cm 0cm 2.5cm}, clip, width=\linewidth]{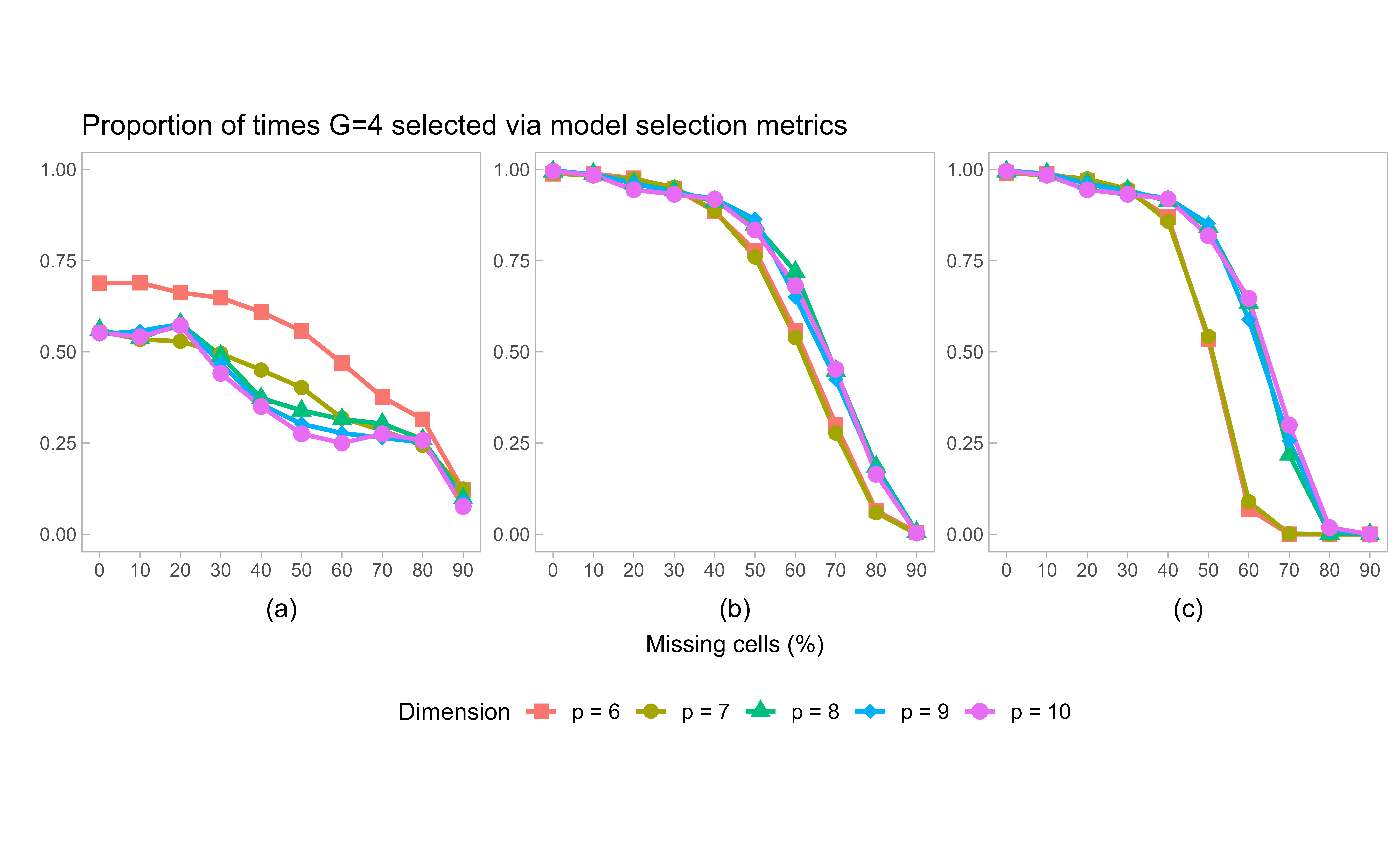}
    \caption{Probability of model selection metric identifying four components present in the dataset: (a) AIC, (b) BIC, and (c) ICL.}
    \label{gof_plot}
\end{figure}

Naturally, higher percentages of values missing corresponds to less accurate results overall. In terms of clustering capabilities, from \figurename \ref{acc_ari_plot} there is a near logistic decay in cluster accuracy and ARI. Recall that an observation belongs to a four component mixture model and thus has a 0.25 probability of being assigned to the correct component at random. Thus, it becomes a lot more challenging for an observation to be correctly classified the more sparse the data becomes. At their extreme of 90\%, the accuracy is 0.32 and the ARI at 0.02. This shows that the chance of accurately classifying an observation is aided by chance when the data is highly sparse. This is also expected since the appearance around the unobserved values is assumed to be independent of the corresponding probability distribution. Crucially, this shows that all observed rows could still be classified despite being incomplete with non-zero accuracy and non-zero ARI.  

\begin{figure}[H]
    \centering
    \includegraphics[trim={0cm 2cm 0cm 2.5cm}, clip, width=\linewidth]{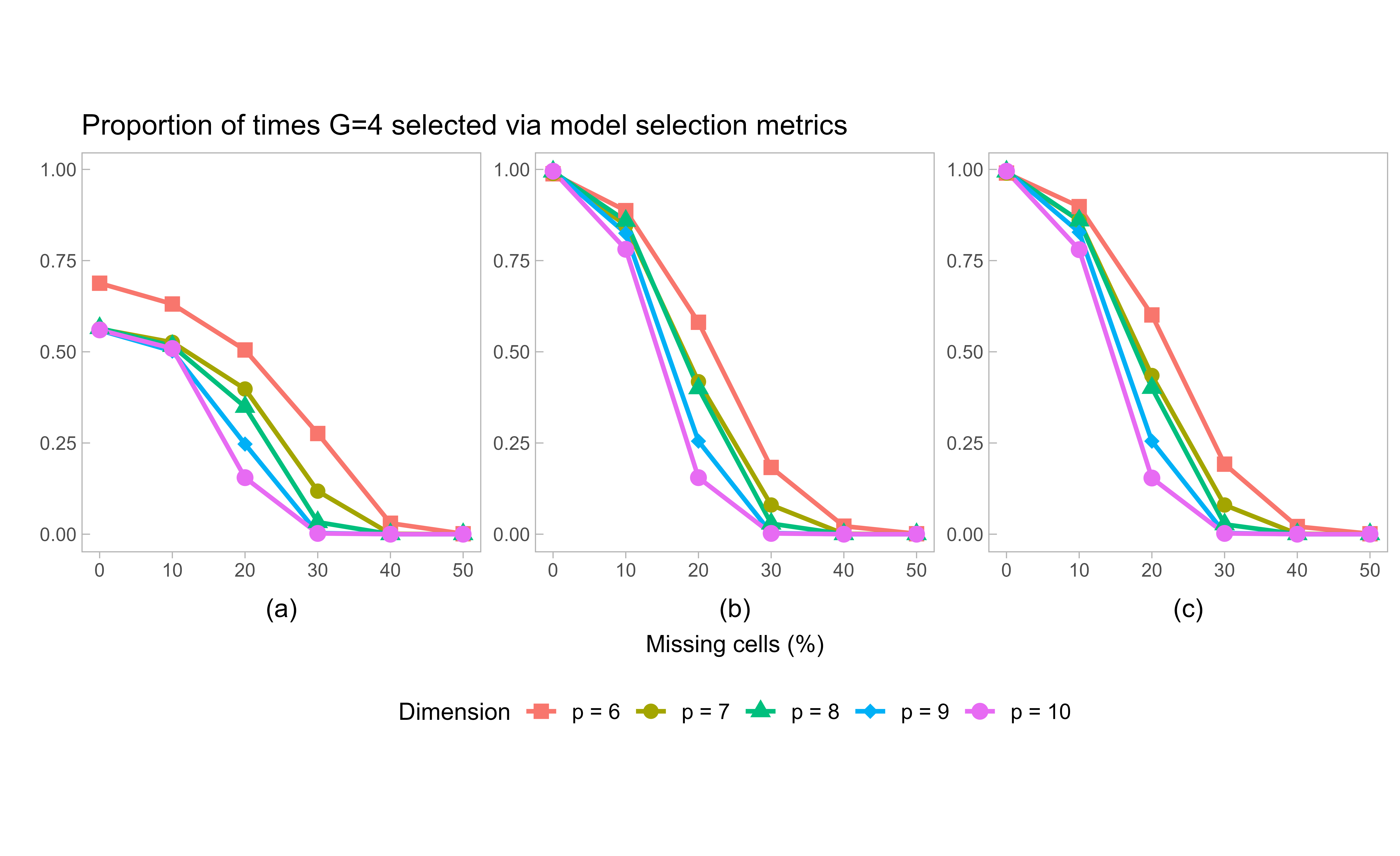}
    \caption{Probability of model selection metric identifying four components present in the dataset when only using a complete subset: (a) AIC, (b) BIC, and (c) ICL.}
    \label{gof_plot_complate_car}
\end{figure}

We now focus on model selection capabilities. In the simulation design, a range of models is fitted on the $b^{th}$ dataset as one would have done in practice, when the number of components are unknown. The range of components as candidates are $1,\dots, p-1$. The reason behind this range is due to the problem with identifiability. In general, mixtures of Dirichlet distributions are not necessarily identifiable. However, it is proven in \citep{identifiability} that a finite mixture of Dirichlet distributions are identifiable when the candidate number of components are strictly fewer than the dimension of the random vector. For this reason, we follow the rule to only fit up to one fewer than the dimension of the observed vector in both simulation experiments and in the data application.

Recording the proportion of datasets for which popular model selection metrics successfully identified the correct number of components, \figurename \ref{gof_plot} displays a surprisingly different trend compared to the clustering performance in \figurename \ref{acc_ari_plot}. Notably, the AIC does not decay from 1, compared to the BIC and ICL metrics. This is to do with the harshness of the penalty inflicted by the metric. It is known that the AIC has a relatively weak penalty on the complexity of a model. Thus, it tends to favour a bigger number of components, which is apparently beneficial for lower dimensional data, namely $p=6$. This leniency also explains why its decay is not as stark as the BIC and ICL measures. 

The BIC and ICL has similar trends. Up to 40\% of entries unobserved, the model selection metrics seem to perform extremely well. A sharp drop in accuracy follows. Notably, these two metrics are more stringent in the penalisation and thus, as the data becomes more sparse, the BIC and ICL tend to favour fewer clusters. Interestingly, the ICL trend forks for more than 40\% of unobserved entries. For fewer dimensions such as $p=6$ and $p=7$ the ICL's penalty is unfittingly harsh, relative to the BIC. This can be explained by the type of penalty the ICL implements. From \eqref{gof_metrics}, the ICL is built from the BIC where another penalty is imposed thereon, based on cluster separability. Hence, it seems, at least for Dirichlet distributed observations, the ICL is best suited for larger dimensions, where visual separation is not as easy to inspect. 

While deletion may not be able to cluster unobserved components, strictly speaking, it can still be used for model selection. Thus, we also compute the proportion of datasets for which a complete subset is used to fit model \eqref{fmm}. The horizontal axes of \figurename \ref{gof_plot_cens_complete} end at 50\% since a four component version of model \eqref{fmm} could not successfully be fitted for datasets with higher percentages of missing values. Notice the difference in the trends relative to \figurename \ref{gof_plot}: in contrast the the excellent performance of BIC and ICL using the proposed algorithm, case deletion lowers the accuracy of these metrics considerably. The probability of these metrics correctly selecting a four component mixture model drops to 0 at only 50\%. 

\subsection{Simualtion results: Part B}

For each dataset generated, values for the first $p-1$ columns are hidden according to the percentages 0\%, 10\%, 20\%, $\dots$, and 90\% for $p=6,\dots, 10$ as in Part A, but under a censoring mechanism. For each percentage of values hidden, the average of the resulting 100 metrics are reported. 

\begin{figure}[H]
    \centering
    \includegraphics[ width=0.8\linewidth]{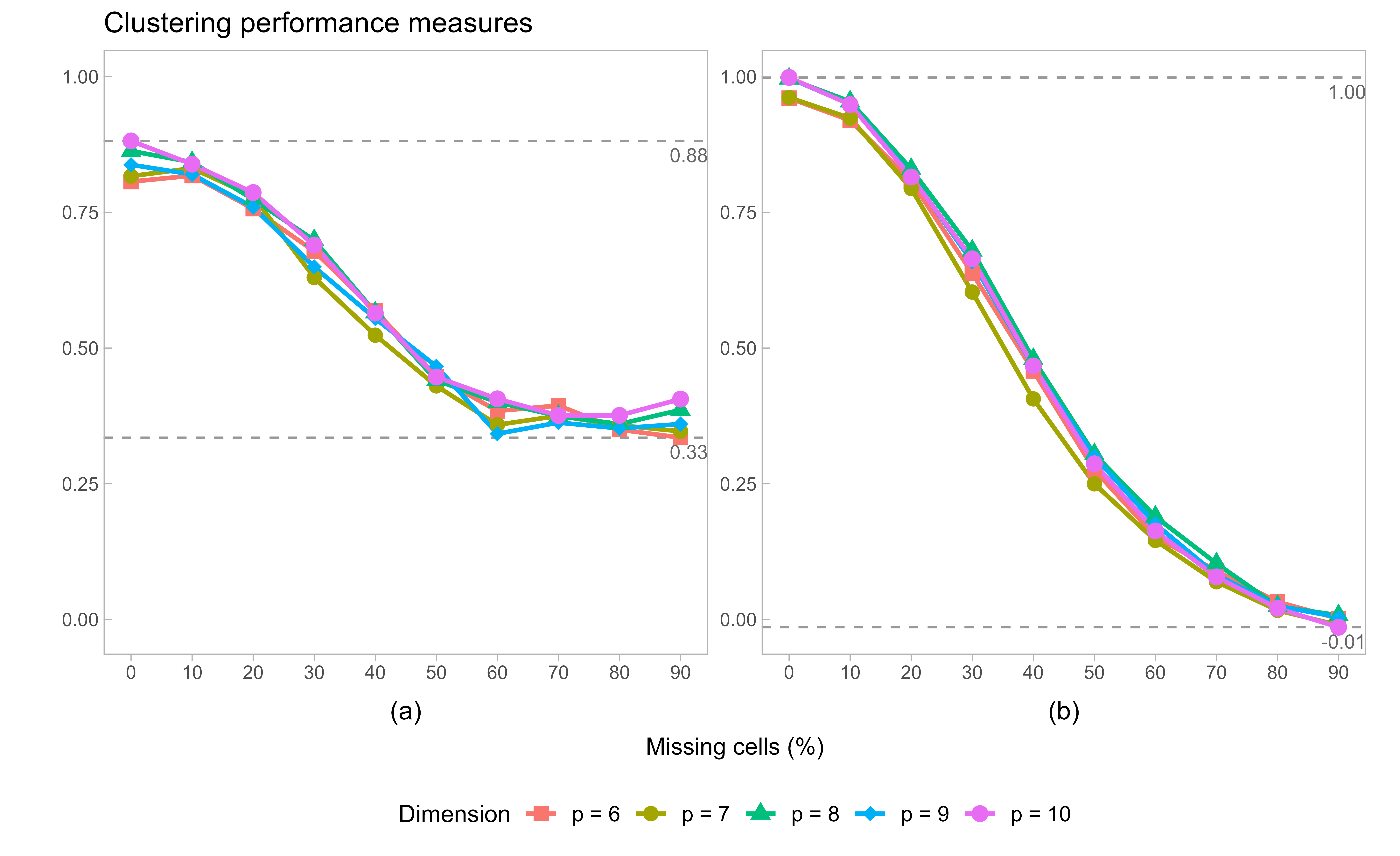}
    \caption{Cluster performance metrics: (a) Average accuracy scores and (b) average ARI scores for simulation Part A.}
    \label{acc_ari_plot_cens}
\end{figure}

The trend exhibited in \figurename \ref{acc_ari_plot_cens} mimics what is seen in \figurename \ref{acc_ari_plot}, even staying in the similar bounds. However, the decay in somewhat faster in the former. This can be explained by the idea that censoring design is a harsher coarsening compared to the MAR design in Part A. However, the outcome from the model fitting process in the simulation experiment is the same. In other words, the proposed algorithm is sufficiently equipped to handle both censored and MAR types of missing values. ARI values can be negative, which occurs when the clustering technique performs worse than random chance. For the extreme case of 90\% of values censored, the average ARI is -0.01. Though, in practical terms, the corresponding average accuracy is at 0.33, suggesting that random chance is fractionally more informative at clustering at best in this case.  
\begin{figure}[H]
    \centering
    \includegraphics[trim={0cm 2cm 0cm 2.5cm}, clip, width=\linewidth]{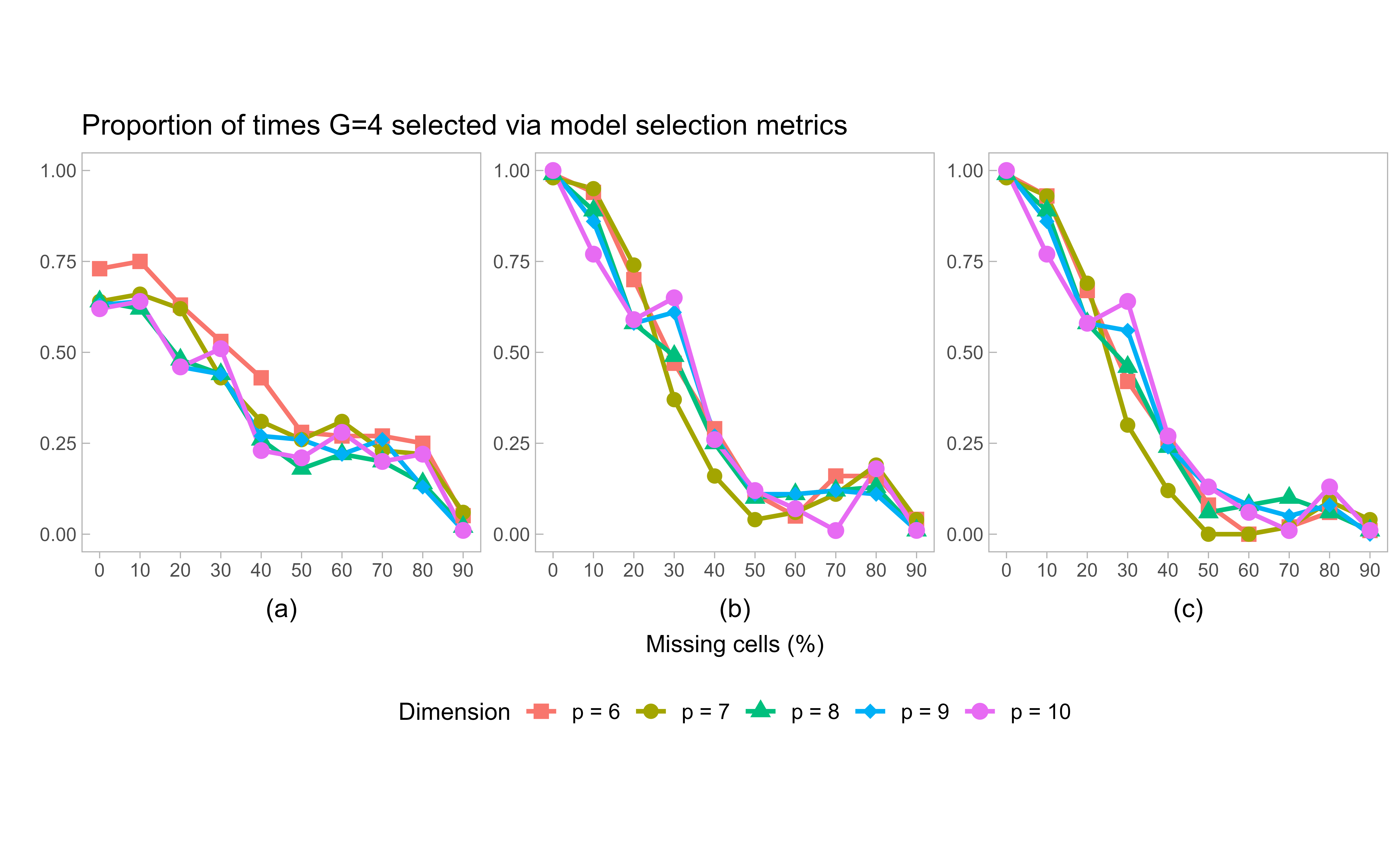}
    \caption{Probability of model selection metric identifying four components present in the dataset: (a) AIC, (b) BIC, and (c) ICL.}
    \label{gof_plot_cens}
\end{figure}

In determining the proportion of times a four component model is selected via selection metrics \eqref{gof_metrics}, \figurename \ref{gof_plot_cens} shows that censoring is a more difficult case to select the correct number of components. Again the decay in probability over an increasing percentage of censored values is more drastic compared to the MAR scenario in Part A. The trend is not as smooth since 100 datasets were generated instead of 1000, simply due to computing constraints. The proposed algorithm requires longer to fit model \eqref{fmm} onto censored data, especially since it is generalised to handle arbitrary feasible regions. That is, it is possible for the proposed algorithm to accommodate each row in a dataset to be subjected to different censored bounds. Despite the steeper decay in \figurename \ref{gof_plot_cens}, the algorithm provides a higher chance of the correct number of components being selected compared to case deletion in \figurename \ref{gof_plot_cens_complete}. At 20\% of censored entries, row deletion has close to a 0 probability of selecting the correct model, thus highlighting the improvement offered by the proposed algorithm.

\begin{figure}[H]
    \centering
    \includegraphics[trim={0cm 2cm 0cm 2.5cm}, clip, width=\linewidth]{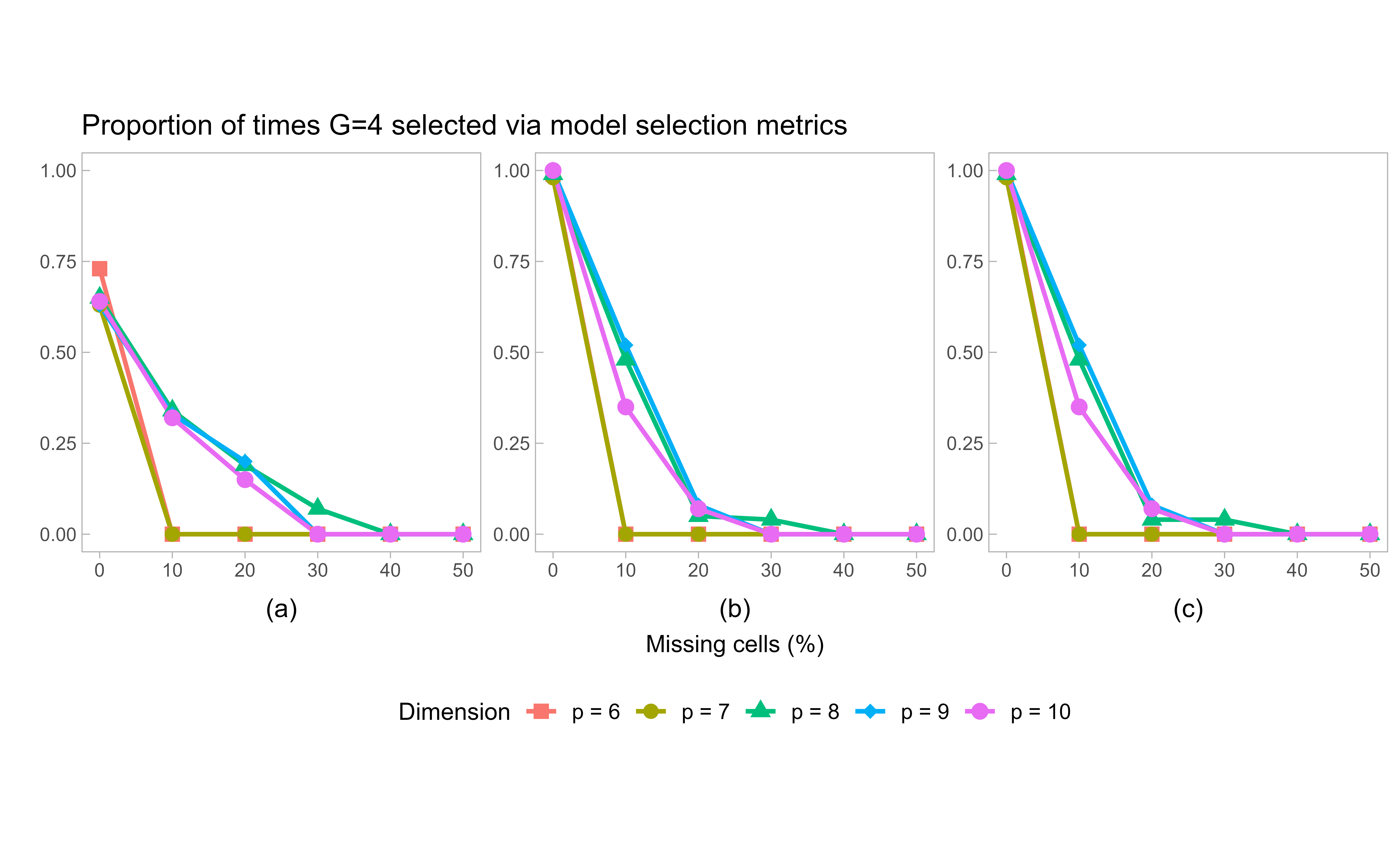}
    \caption{Probability of model selection metric identifying four components present in the dataset when using the complete subset: (a) AIC, (b) BIC, and (c) ICL.}
    \label{gof_plot_cens_complete}
\end{figure}

\section{Application}
\label{application}

The applications demonstrate the usefulness of the algorithm when applying the finite mixtures of Dirichlet to gather insights on compositional data, and its effectiveness at identifying clusters within incomplete data. It considers two datasets: A compiled EarthChem geochemistry dataset and the Air Quality Systems dataset. Both datasets exhibit special types of CAR mechanisms. %The former has values missing at random and the latter has a mix of values that are censored and are missing at random.

\subsection{EarthChem geochemistry data}
EarthChem is a community-driven data respository that promotes the accessibility, interoperability, and reuse of geochemical data by aggregating several major geochemical databases, including PetDB, GEOROC, and the USGS National Geochemical Database. It offers a domain-specific "Data-to-go" compilation derived from PetDB, a curated database of geochemical measurements from igneous and metamorphic rocks compiled primarily from the published scientific literature. The mantle xenoliths dataset is used to illustrate the estimation method discussed in Section \ref{inference}.

The selected dataset consists of major oxide compositions, trace element concentrations, isotope measurements, and associated metadata for mantle xenolith samples. Mantle xenoliths are fragments of the Earth's mantle transported to the surface by volcanic activity and provide direct insight into the composition and evolution of the upper mantle. Consequently, the chemical classification of these rocks are used to understand mantle processes, tectonic evolution, and the formation of the Earth's lithosphere \citep{xenolith}.

This dataset is particularly well suited to the proposed algorithm because the geochemical measurements are compositional, with oxide concentrations representing parts of a whole and therefore constrained to sum to unity after some data cleaning. Furthermore, geochemical datasets commonly exhibit heterogeneity arising from the wide range of factors, such as place, age, and disruption, making them suitable for applications of finite mixture models.  The accompanying metadata, including rock type and tectonic setting, are not used during model fitting but provide an independent means of assessing the geological agreement and interpretability of the resulting clusters. \footnote[1]{PetDB Team, T., 2019. EarthChem Data-To-Go: Geochemical Data for Mantle Xenoliths, version February 2019, Version 1.0. Interdisciplinary Earth Data Alliance (IEDA). https://doi.org/10.1594/IEDA/111309. Accessed 29-06-2026.}
%\begin{table}[ht]
%    \centering
%    \caption{Percentage of unobserved values by molecule.}
%    \begin{tabular}{lSSSSSSSS}
%        \toprule
%        & {$\text{SiO}_2$} & {$\text{Al}_2\text{O}_3$} & {CaO} & {$\text{Cr}_2\text{O}_3$} & {$\text{Fe}_2\text{O}_{3\mathrm{T}}$} & {MgO} & {$\text{Na}_2\text{O}_2$} & {Residual} \\
%        \midrule
%        Missing (\%) & 4.505929 & 1.185771 & 4.743083 & 44.426877 & 83.399209 & 3.003953 & 27.826087 & 96.758893 \\
%        \bottomrule
%    \end{tabular}
%    \label{missing_earthchem}
%\end{table}
This dataset contains major oxides, trace element, and isotope compositional data of all Mantle Xenolith samples including reference information and metadata such as analytical method used to speciate the data, the sample type, the expidition it was collected, and the article it was used in. For this application, we focus on compilation of the major oxides. Aggregation on such a large scale undoubtedly introduces missing values. After filtering to the X-ray Fluorescence (XRF) method of speciation, and removing problematic rows with measurements that were not constrained to unity, the dataset consisted of $n=1256$ rows, $p=8$ major oxides, and with an overall percentage of cells unobserved at 33.2312\%.

\begin{figure}[H]
    \centering
    \includegraphics[width=0.6\linewidth]{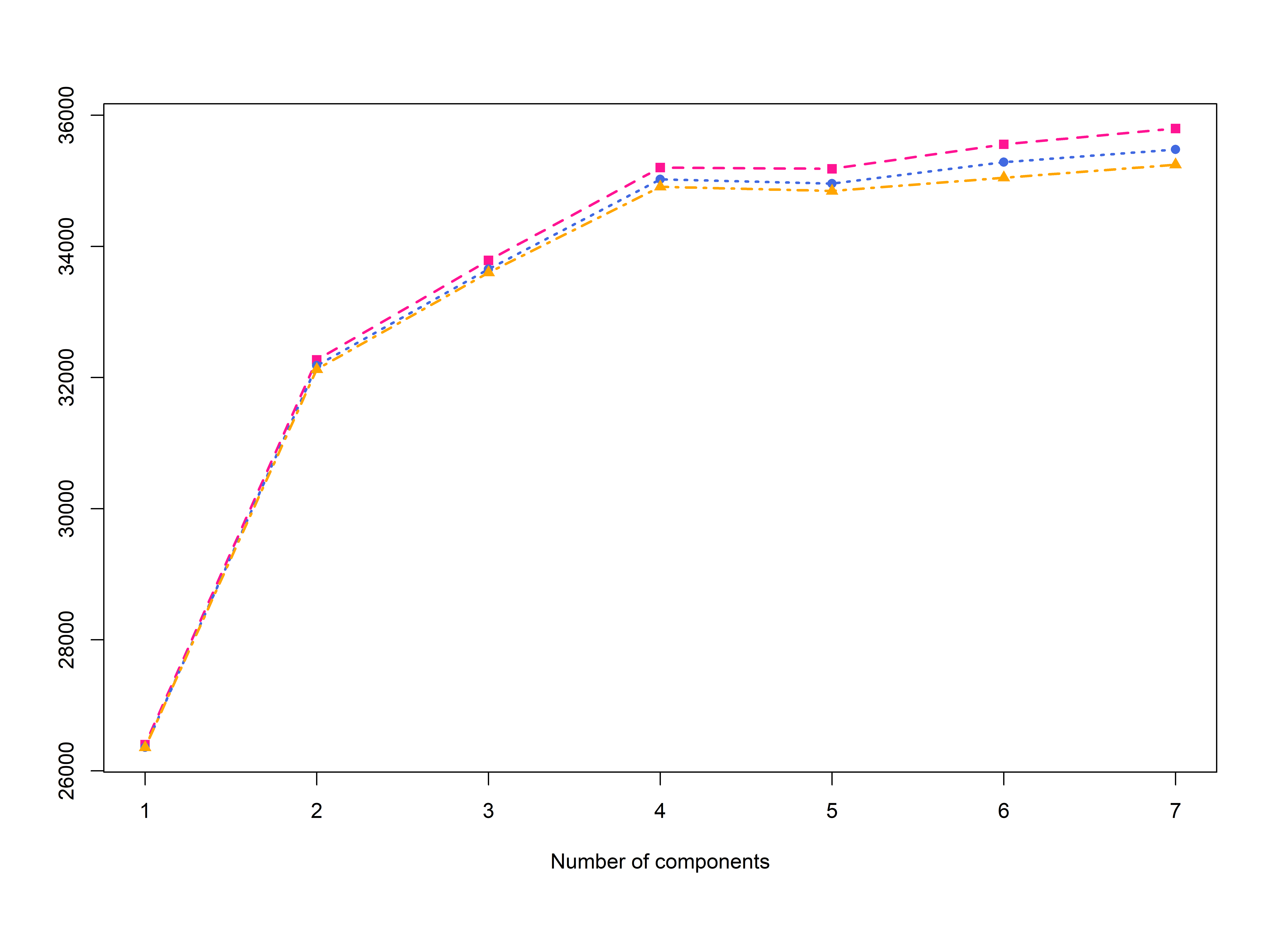}
    \caption{Model selection metrics for model \eqref{fmm} fitted for candidate components $G=1,\dots,7$. }
    \label{model_selection_plot}
\end{figure}

After fitting model \eqref{fmm} with $G=1,\dots,7$ components, the model selection criteria in \figurename \ref{model_selection_plot} display an interesting pattern. As anticipated from the simulation study, the AIC favours the seven-component model. Although both the BIC and ICL also obtain their maximum values at G=7, the increase in the ICL beyond four components is noticeably more subdued. This behaviour is consistent with the increasing strength of the complexity penalty, which is smallest for the AIC, followed by the BIC, and largest for the ICL. 

Focusing on the increase in these metrics, shows a bend at $G=4$ components is present. Inspection of the parameter estimates for the models with $G=5,6$, and $7$, shows that, beyond the four largest components, the estimated mixing probabilities are all of order $10^{-5}$ or smaller. That is, these additional components are of negligible size, while their corresponding density parameters become extremely large. The resulting component densities are therefore highly concentrated around only a handful of observations, producing peaks that inflate the log-likelihood. This is illustrated in \figurename \ref{parameter_norms}: The added natural logarithm of the estimated parameters' 1-norms, shown in \figurename~\ref{alpha_norms}, increases sharply once more than four components are fitted. For $G>4$, the parameter magnitudes more than double before eventually reaching a plateau due to the limits of machine precision. Simultaneously, \figurename~\ref{prob_norms} shows that the mixing probability corresponding to the largest estimated parameter vector, $\balpha$ decreases sharply, indicating that these increasingly extreme parameter estimates are associated with components of negligible size. This behaviour does put into question whether typical model selection metrics are interpretable for distributions on the simplex whose concentration can theoretically be arbitrarily large. Indeed, looking as previous works on mixtures of Dirichlet distributions, \cite{dirichlet_nested_fmm}, instead considers a dynamic algorithm that eliminates a component once it's mixing proportion has fallen below a certain threshold. A similar approach can be taken in future work.

\begin{figure}[H]
    \centering
    \subfloat[]{%
        \includegraphics[width=0.45\textwidth]{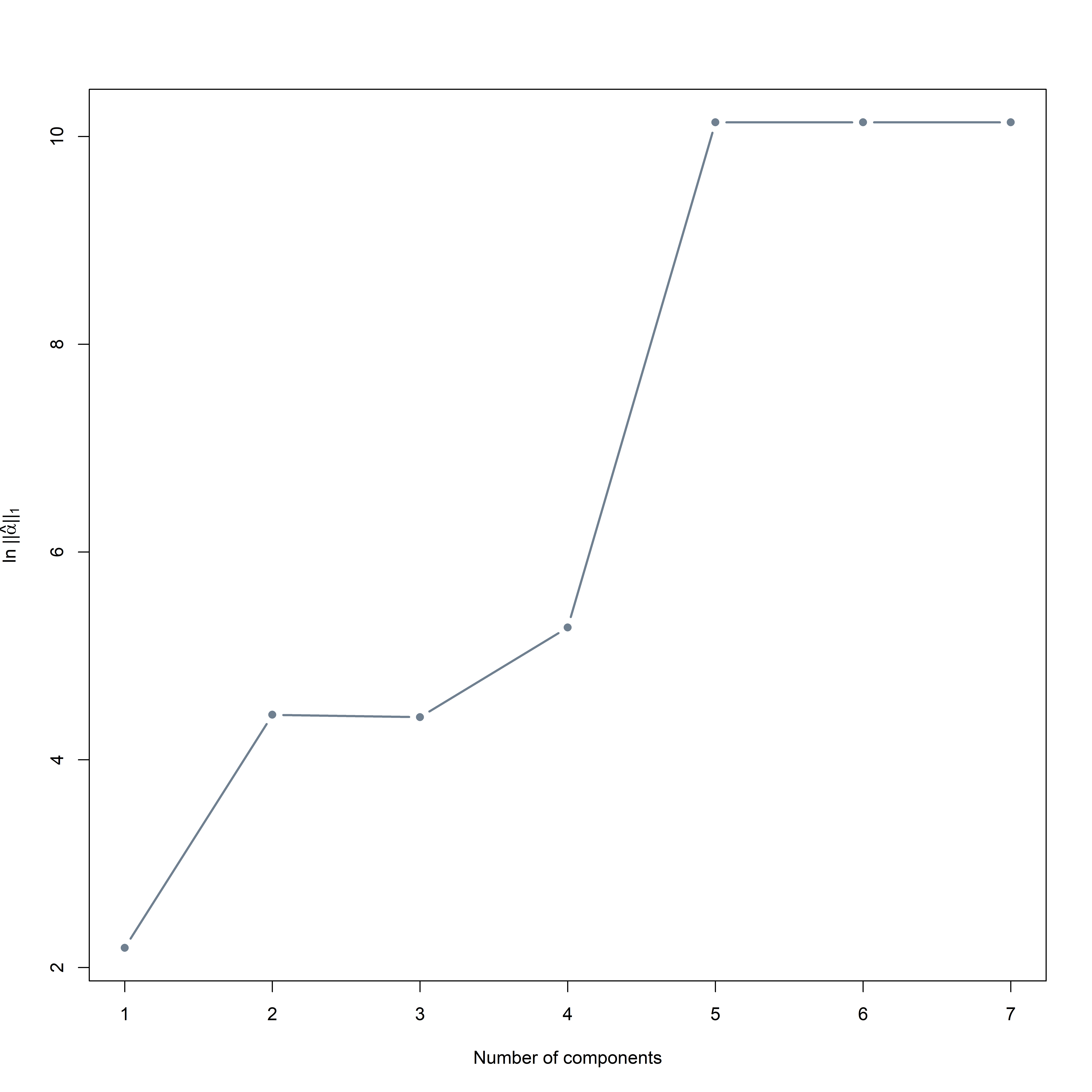}
        \label{alpha_norms}
    }
    \hfill
    \subfloat[]{%
        \includegraphics[width=0.45\textwidth]{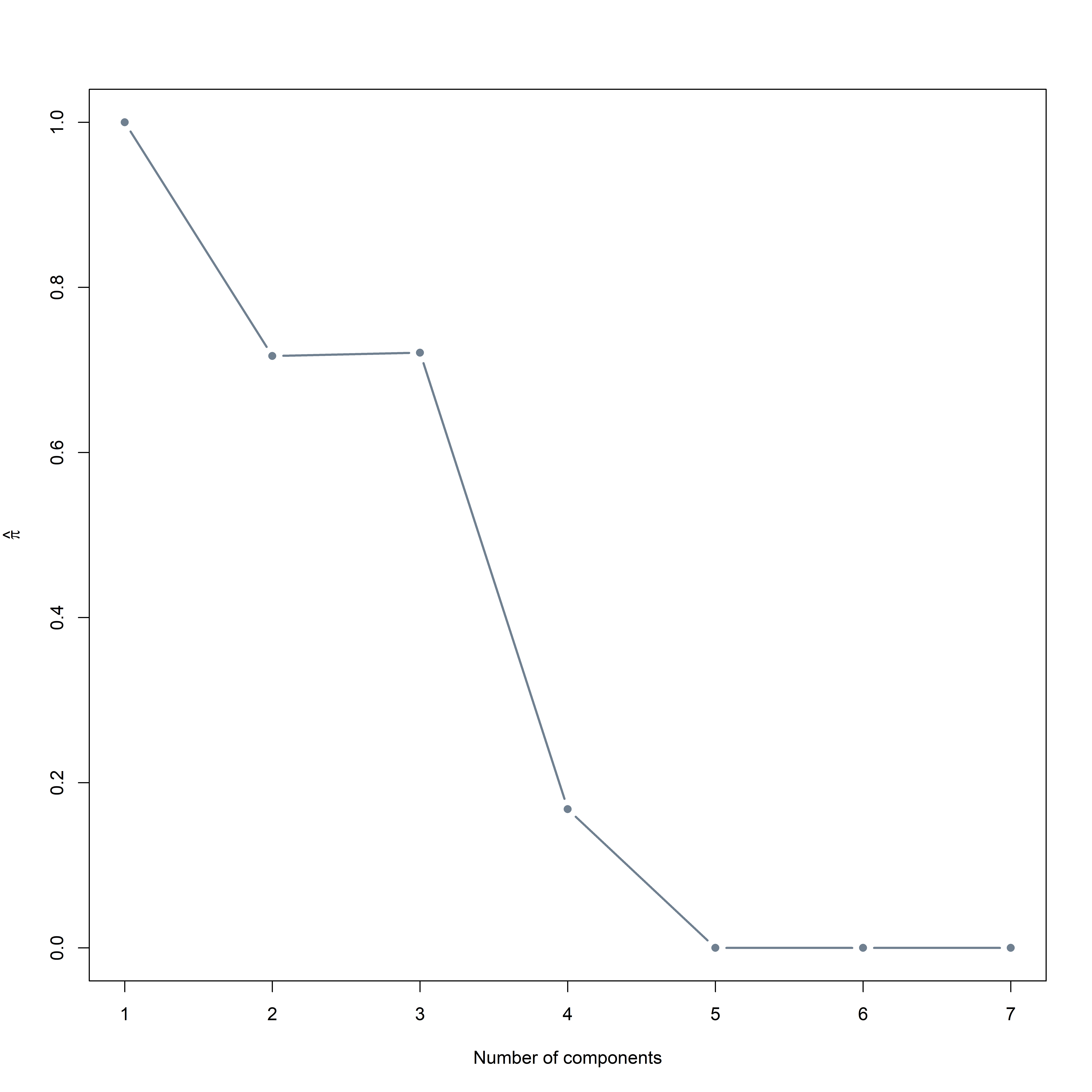}
        \label{prob_norms}
    }
    \caption{1-norm distances of parameter estimated (in natural logarithm for scale) (a) and the value of the mixing probability corresponding the largest estimated parameter (b), respectively. }
    \label{parameter_norms}
\end{figure}
Another interesting artefact of this application is comparing the clustering results for 3 components against 4 components. \figurename \ref{comparison_3_4} displays the estimated clusterwise-average elemental composition. Notice that the behaviours of the green, blue, and red lines are almost the same across the two models. In particular, the additional component in \figurename \ref{4_comp}, displayed as an orange line, has the exact same behaviour as the blue line, with the exception of the residual part. Thus, the interpretation of three clusters remain the same across the models, with the additional cluster placing emphasis on the proportion of unaccounted molecules.
\begin{figure}[H]
    \centering
    \subfloat[]{%
        \includegraphics[width=0.45\textwidth]{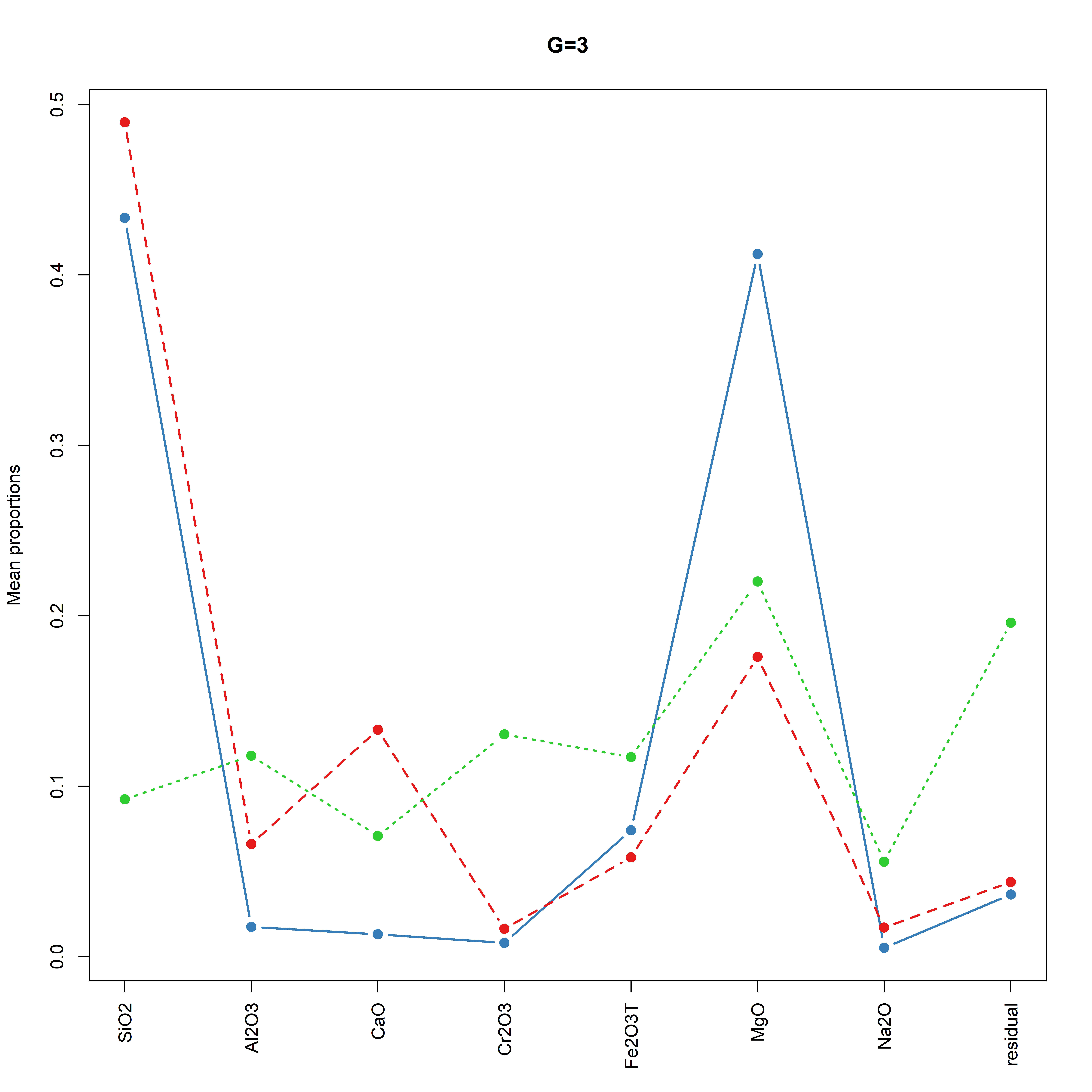}
    }
    \hfill
    \subfloat[]{%
        \includegraphics[width=0.45\textwidth]{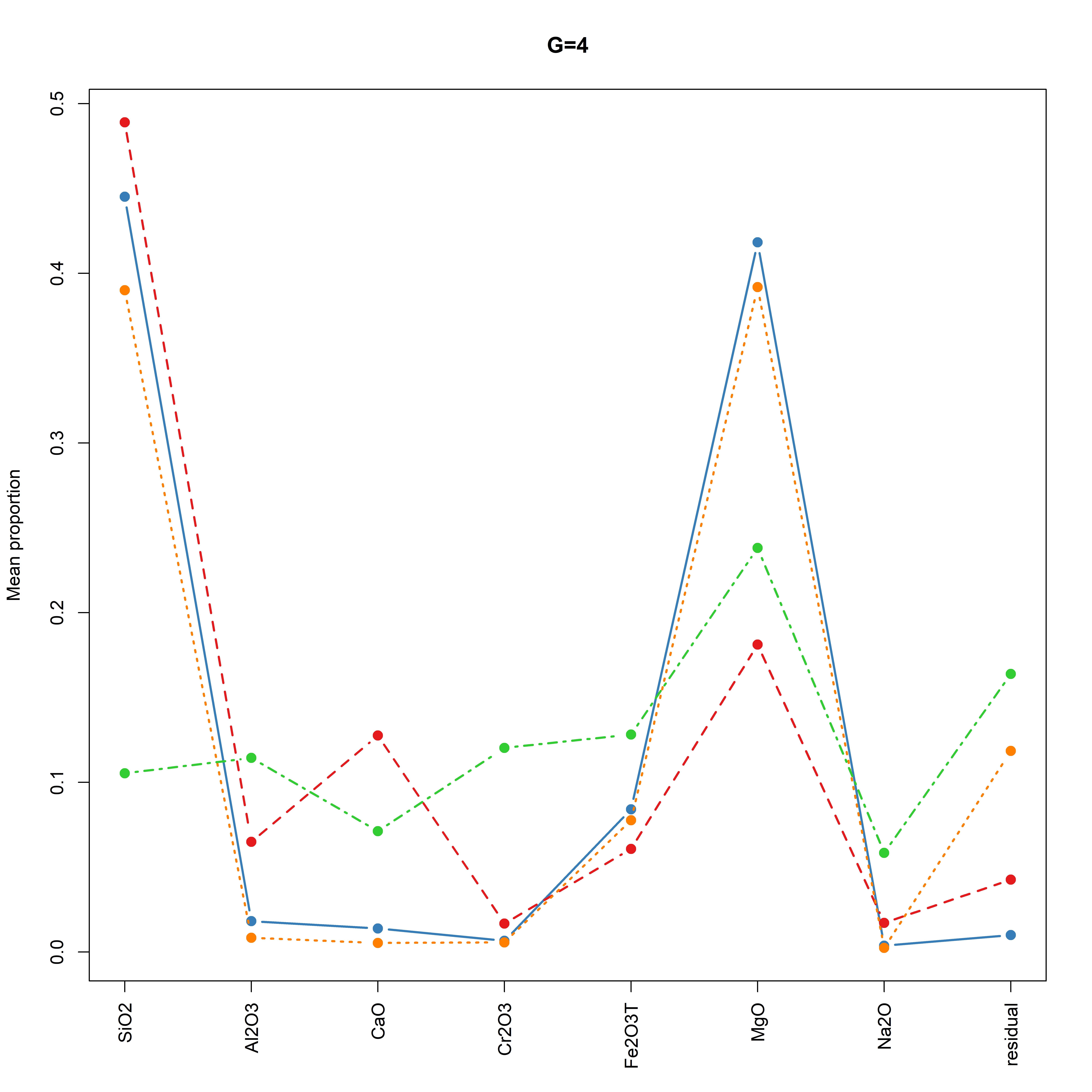}
        \label{4_comp}
    }
    \caption{Mean proportions by component for three components (a) and four components (b), respectively. }
    \label{comparison_3_4}
\end{figure}

Thus, we consider an exploration of the clusters under the 4-component fit. The estimated parameters are given in \tablename \ref{mle_estimates_earthchem}.
\begin{table}[ht]
    \centering
    \caption{Estimated parameters and mixing probabilities for the four-component Dirichlet mixture model fitted to the EarthChem dataset.}
    \setlength{\tabcolsep}{1cm}
    \begin{tabular}{lSSSS}
        \toprule
        Parameter & \multicolumn{4}{c}{Component} \\
        \cmidrule(lr){2-5}
                              & {$g=1$} & {$g=2$} & {$g=3$} & {$g=4$} \\
        \midrule
        $\hat{\bm{\pi}}$                        & 0.5349937  & 0.1856603  & 0.1114633 & 0.1678827  \\
        $\mathrm{SiO}_2$                        & 64.0041579 & 15.3372091 & 0.3229760 & 76.1231797 \\
        $\mathrm{Al}_2\mathrm{O}_3$             & 2.6181348  & 2.0365788  & 0.3506761 & 1.6382374  \\
        $\mathrm{CaO}$                          & 1.9982953  & 4.0014601  & 0.2183002 & 1.0436875  \\
        $\mathrm{Cr}_2\mathrm{O}_3$             & 0.9510733  & 0.5244204  & 0.3688581 & 1.1104286  \\
        $\mathrm{Fe}_2\mathrm{O}_{3\mathrm{T}}$ & 12.0847763 & 1.9065954  & 0.3929960 & 15.1497343 \\
        $\mathrm{MgO}$                          & 60.1331750 & 5.6818076  & 0.7301035 & 76.4747314 \\
        $\mathrm{Na}_2\mathrm{O}$               & 0.5398504  & 0.5387984  & 0.1791868 & 0.4821384  \\
        residual                                & 1.3369598  & 0.5025491  & 1.4397990 & 23.1276094 \\
        \bottomrule
    \end{tabular}
    \label{mle_estimates_earthchem}
\end{table}

The dataset also has descriptor columns, including the tectonic setting where the sample was taken and the name of the rock the sample originates from, which were not used in the clustering process. It would aid interpretation to view how the labels were assigned. Since there are numerous tectonic settings and rock names, a bar plot is not sufficient. Thus, \figurename \ref{word_cloud_4} has word clouds per cluster of the rock's name and place. The clusters can thus be desribed as follows:
Cluster 1: 
represents the dominant mantle-derived peridotites with relatively well-characterised major oxide compositions, having the highest mixing probability. The cluster is concentrated around a composition rich in  $\text{SiO}_2$ and MgO and is primarily associated with peridotite samples collected from cratonic tectonic settings.
Cluster 2: Rocks of this cluster are characterised by higher proportions CaO compared to other clusters. This cluster also groups Pyroxenite rocks sampled in intraplate off-craton settings.\\
Cluster 3: This cluster is distinguished by a substantially lower proportion of $\mathrm{SiO}_2$ than the remaining clusters, together with elevated proportions of $\mathrm{Fe}_2\mathrm{O}_{3\mathrm{T}}$, $\mathrm{Cr}_2\mathrm{O}_3$, $\mathrm{Al}_2\mathrm{O}_3$, and the largest residual proportion. Unlike the other clusters, no single oxide from Harzburgitic samples dominates apart from MgO, while the residual proportion is considerably larger, suggesting that a substantial fraction of the composition lies outside the measured major oxides. The comparatively small total parameter vector also indicates substantially greater compositional variability within this cluster.\\
Cluster 4: As seen in \figurename \ref{4_comp}, the compositional behaviour is similar to cluster 1, but with a larger proportion left unidentified. This is further motivated by the word cloud in \figurename \ref{word_cloud_4}, displaying the major rock type and tectonic setting agreeing with those of cluster 1. However, this cluster also sees samples from Ophiolitic settings than cluster 1, which is largely craton specific. From this cluster it is suspected that the clustering separates chemically similar rock types according to subtle differences in composition and geological origin.

Notice that the geological descriptors were not included in the clustering algorithm, the resulting groups relay clear geological labels. Samples assigned to the same mixture component are thus associated with similar rock types and tectonic environments.
\begin{figure}[H]
    \centering
    \includegraphics[width=0.75\linewidth]{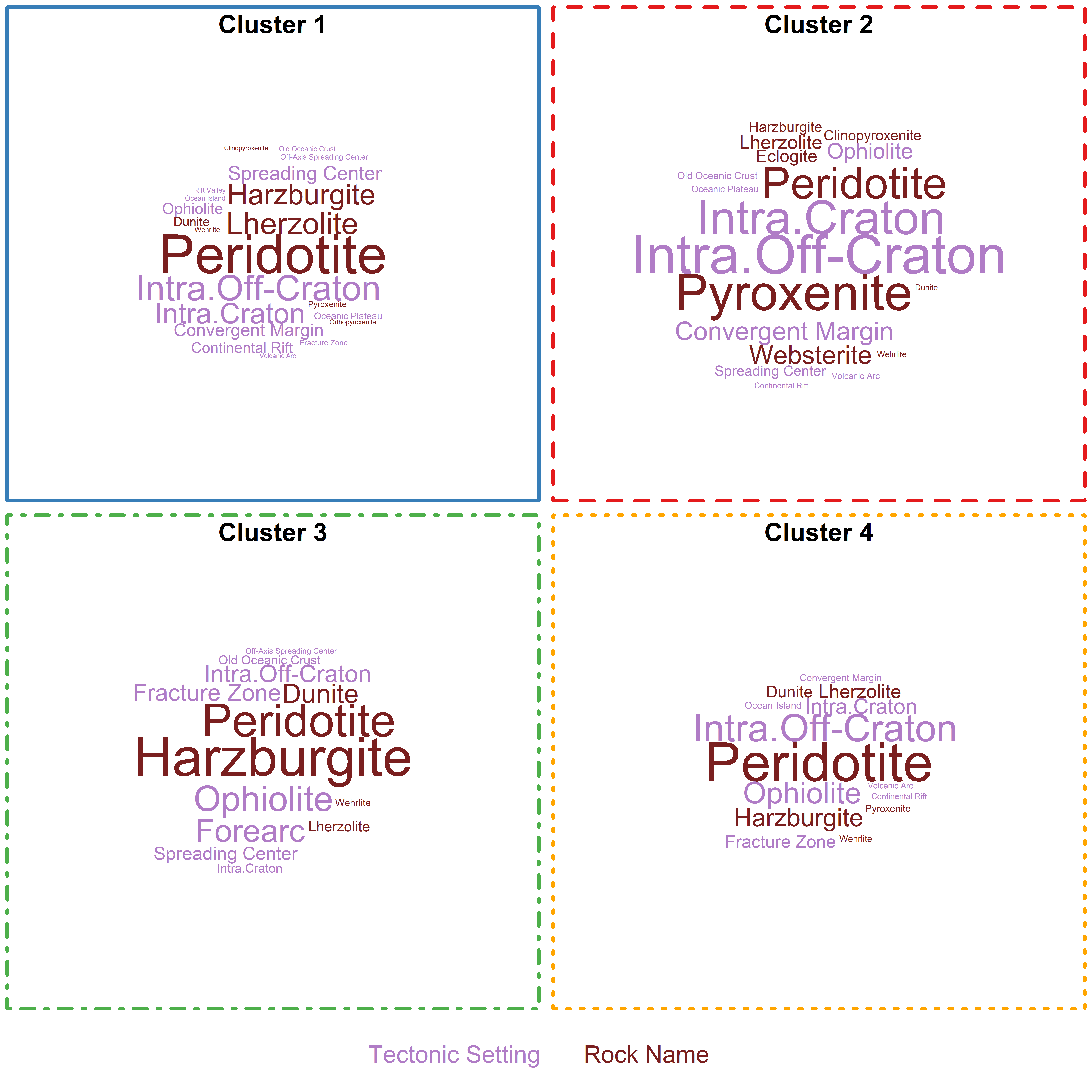}
    \caption{Word clouds of rock samples after assigning the dataset's rows to the cluster determined by fitting a four component Dirichlet mixture model.}
    \label{word_cloud_4}
\end{figure}

\subsection{Air Quality System data}
The Air Quality System (AQS) is a data repository that contains ambient air pollution data collected by the Environmental Protection Agency from monitors spread across the United States to support evaluation of long-term trends and to better quantify the impact of species of particulate matter (PM) concentrations in the size range below 2.5 mm aerodynamic diameter, among other particulates \citep{AQS_reason}. Samples collected from the air through monitors are tested for numerous molecule species linked to meteorology, pollution, and toxicity. Historic measurements can be accessed through the AQS website. \footnote[3]{US Environmental Protection Agency. Air Quality System Data Mart [internet database] available via \url{https://www.epa.gov/outdoor-air-quality-data}. Accessed 30 January 2026.} Air polluted with high levels of PM2.5 is a global concern. PM suspended in the atmosphere significantly alters air quality, influences climate change, and injurious to human health. Human exposure to these fine particles, which penetrate deep in the lungs, causes respiratory and cardiovascular diseases, non-infectious diseases, and a lowered quality of life. In addition to health effects, particulate pollution leads to haze formation and decreases visibility. PM2.5 particles contain various chemical components such as  organic carbon and nitrate $(\text{NO}_3^-)$ and sulphate $(\text{SO}_4^{2-})$ ions. The presence and relative proportions of these components provide insight into the dominant emission aerosols affecting a region. Understanding the speciation of PM2.5 and the factors that affect its composition provide supporting information leading to improved control measures and policies to reduce the impact of air pollution. In other words because different species have different pathways of being formed (i.e. several chemical reactions that lead to the same species in the air), chemical and physical effects, and lifespans, their corresponding control strategies will also differ. Thus, knowing their compositions in air can help determine what interventions are most effective. 

The AQS website hosts PM2.5 daily concentrations and their corresponding speciations. We apply the proposed algorithm on a snapshot of the data, $8^{th}$ March 2025. From prior studies on PM2.5 compositions \citep{apportionment, compositions, compositions2007,compositions2012}, commonly used variable include: ammonium ions ($\text{NH}_4^+$), elemental carbon emissions $(\text{EC})$, nitrates, $(\text{NO}_3^-)$, organic carbon emissions (OC), potassium ions $(\text{K}^+)$, silicon $(\text{Si})$, sodium ions $(\text{Na}^+)$, and sulphates $(\text{SO}_4^{2-})$. Monitors are subject to physical limitations, known as method detection limits (MDLs), meaning that while molecules might be detected, its actual concentration falls below a physical threshold and is not calculable. That is, the measurements are subject to left-censoring.  A range of monitors are designed for different groups of molecules and thus the MDLs are varied across molecules. From the AQS code list, the following MDLs for the respective molecules are: $< 0.017, <0.002, <0.008, <0.002, <0.014, <0.00753, <0.030$, and $<0.012$. Furthermore, not all locations took measurements for all molecules for the day. The resultant dataset consists of $n=99$ rows and suffers from both values missing at random (MAR) and censored observations, that is covered under the CAR mechanism. 
\begin{table}[ht]
    \centering
    \setlength{\tabcolsep}{1cm}
    \caption{Percentage of unobserved values by molecule.}
    \begin{tabular}{lSSl}
    \toprule
       Molecule           & {Missing (\%)}      &  {Censored (\%)} & MDL threshold\\
       \midrule  
       $\text{NH}_4^+$    & 7.070707            &   4.040404       &  $<$ 0.017   \\
       EC                 & 10.101010           &   1.010101       &  $<$ 0.002   \\
       $\text{NO}_3^-$    & 3.030303            &   0              &  $<$ 0.008   \\
       OC                 & 7.070707            &   0              &  $<$ 0.002   \\
       $\text{K}^+$       & 24.242424           &   21.212121      &  $<$ 0.014   \\
       $\text{Si}$        & 5.050505            &   3.030303       &  $<$ 0.00753 \\
       $\text{Na}^+$      & 38.383838           &   35.353535      &  $<$ 0.030   \\
       $\text{SO}_4^{2-}$ & 3.030303            &   0              &  $<$ 0.012   \\
        \bottomrule
    \end{tabular}
    \label{missing_aqs}
\end{table}

Model \eqref{fmm} was fitted on the dataset using the proposed algorithm for $K = 1,\dots,8$. A 4-component mixture model produced the highest ICL value out of the 8 fitted.

The estimated cluster proportions and means are given in \tablename~\ref{speciation_mle}:

\begin{table}[ht]
    \centering
    \setlength{\tabcolsep}{1cm}
    \caption{Estimated cluster proportions and averages for the AQS dataset, given per cluster. The estimated averages indicate what the average composition of molecules per cluster.}
    \begin{tabular}{lSSSS}
    \toprule
     Statistic  & \multicolumn{4}{c}{Cluster} \\
                  \cmidrule(lr){2-5}
                & {$g=1$} & {$g=2$} & {$g=3$} & {$g=4$} \\
            \midrule  
$\hat{\bm{\pi}}$    & 0.09882622  & 0.14545621   & 0.17222936 & 0.58348821 \\
$\text{NH}_4^+$     & 0.07420415  & 0.039732040  & 0.01875825 & 0.019398241 \\
EC                  & 0.03414402  & 0.046908328  & 0.02083747 & 0.042635363 \\
$\text{NO}_3^-$     & 0.34390680  & 0.216541059  & 0.06019464 & 0.054994856 \\
OC                  & 0.20576989  & 0.244681144  & 0.11100912 & 0.262431175 \\
$\text{K}^+$        & 0.01178039  & 0.008483184  & 0.01046683 & 0.007245226 \\
$\text{Si}$         & 0.01173320  & 0.014148719  & 0.02599546 & 0.017548164 \\
$\text{Na}^+$       & 0.01385559  & 0.005195346  & 0.02585452 & 0.012616918 \\
$\text{SO}_4^{2-}$  & 0.17238864  & 0.141834394  & 0.12205813 & 0.129140010 \\
Residual            & 0.13221732  & 0.282475786  & 0.60482557 & 0.453990047 \\
        \bottomrule
    \end{tabular}
    \label{speciation_mle}
\end{table}

\begin{figure}[H]
    \centering
    \includegraphics[width=0.8\linewidth]{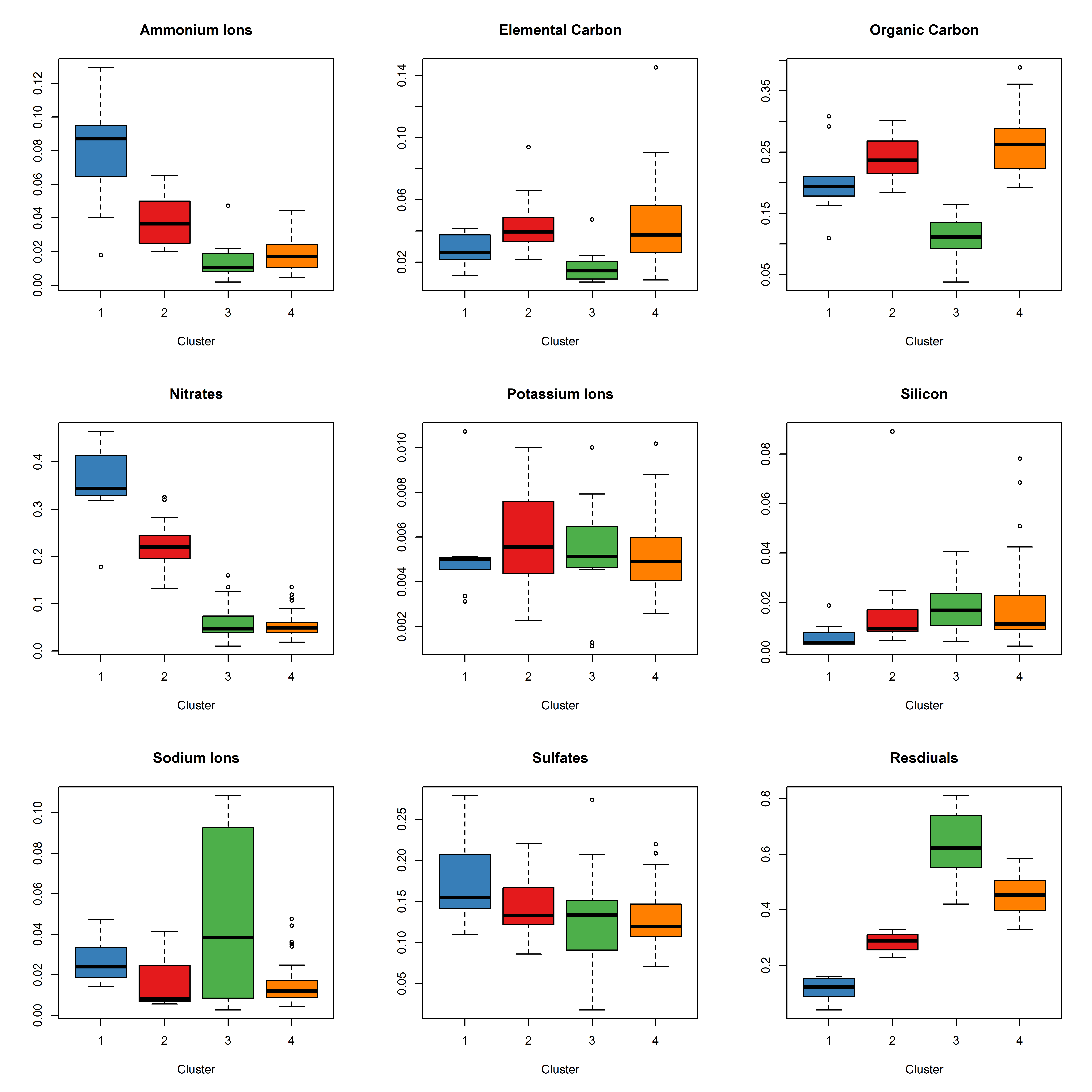}
    \caption{Cluster-wise box plots of the molecule proportions from the AQS dataset clustered according to the MAP classifications after fitting a $K=4$-component Dirichlet mixture model.}
    \label{speciation}
\end{figure}

\begin{figure}[H]
    \centering
    \includegraphics[width=0.7\linewidth]{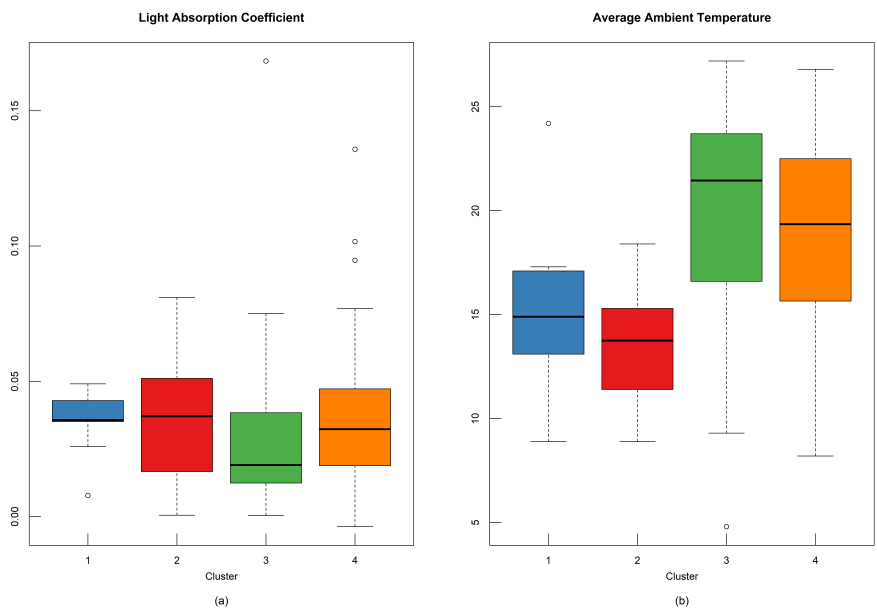}
    \caption{Useful additional variables not used in the clustering algorithm to aid context of the identified clusters. The light absorption coefficient (per 100 megametres) in (a) and the average temperature on the day the molecules were measured in (b).}
    \label{absorp_temp}
\end{figure}
%π 0.134 0.133 0.568 0.165
%NH+
%4 0.032 0.067 0.019 0.019
%EC 0.049 0.037 0.043 0.020
%NO−
%3 0.259 0.213 0.261 0.110
%OC 0.186 0.324 0.053 0.058
%K+ 0.007 0.012 0.007 0.011
%Si 0.018 0.011 0.017 0.026
%Na+ 0.001 0.013 0.013 0.027
%SO2−
%4 0.130 0.170 0.130 0.123
%Residual 0.317 0.152 0.456 0.605
Cluster 1:  Fewer than 10\% of observations fall within cluster one, which is characterised by relatively high proportions of ammonium ($\text{NH}_4^+$), nitrate $(\text{NO}_3^-)$, and sulphate $(\text{SO}_4^{2-})$ compared with the other clusters. This cluster also exhibits lower contributions from potassium $(\text{K}^+)$ and silicon (S) ions. The measured chemical species provide a comparatively complete description of PM2.5 composition for this cluster, as indicated by its small residual part, as seen in \figurename~\ref{speciation}. Observations in this group are associated with lower ambient temperatures relative to the remaining clusters. Higher sulphate $(\text{SO}_4^{2-})$ proportions in particulate matter are widely understood to arise from sulphur-containing gases in the atmosphere that are transformed into particulate sulphate $(\text{SO}_4^{2-})$\citep{sulphate_formation, sulphate_duration, s02_ships}. Higher sulphate $(\text{SO}_4^{2-})$ proportions thus indicate the influence of sulphur-related emissions in the sampled air mass.  Similarly, higher proportions of ammonium ($\text{NH}_4^+$) and nitrate $(\text{NO}_3^-)$ in particulate form are commonly observed under cooler ambient conditions, whereas warmer conditions favour a reduced particulate contribution of these species \citep{ammonium_nitrate_volatile}. This pattern is consistent with the lower ambient temperatures observed for clusters one and two in Figure \ref{absorp_temp}.

In cluster two: PM2.5 is largely composed of elemental carbon, organic carbon, and nitrate $(\text{NO}_3^-)$. This cluster also shows relatively high variability in potassium $(\text{K}^+)$ ion contributions. Overall, this cluster reflects a mixed composition in which both carbon-based particulates and inorganic ions contribute substantially to PM2.5 mass. The co-occurrence of these components suggests conditions under which multiple particle constituents are present simultaneously, rather than dominance by a single class of compounds. Such mixed compositions have been observed in previous studies of urban PM2.5, where particles reflect contributions from a range of emission-related and atmospheric processes rather than a single dominant source.\newline
\newline
Cluster 3: Characterised by a large residual part, this cluster indicates that a significant portion of PM2.5 mass is not captured by the measured chemical species included in this dataset. These measurements were taken during higher ambient temperatures and show lower light absorption, suggesting that the unmeasured molecules are less light-absorbing than in other clusters. The combination of high residual mass and low light absorption suggests that the unmeasured molecules contributing to PM2.5 in this cluster are predominantly weakly light-absorbing. This contrasts with clusters exhibiting higher absorption coefficients, where light-absorbing ions comprise a larger fraction of PM2.5 mass. These findings indicate that, under warmer conditions, PM2.5 composition is less fully described by ion and carbon measurements. This signals the presence of additional, distinct particles, and thus warrants further investigation. Such particles are commonly darker in appearance and absorb more visible light than the weakly absorbing material inferred for cluster 3. Well-known examples of strongly light-absorbing particulate matter include soot-like particles produced during combustion processes, such as fires and diesel exhaust from traffic, are strongly light-absorbing and contribute substantially to measured higher light absorption coefficients \citep{oc_lac}.
\newline
\newline
Cluster 4: This cluster contains the majority of observations. Readings from this cluster have the largest relative  proportion of organic carbon alongside comparatively low proportions of nitrate $(\text{NO}_3^-)$ and other inorganic ions. Observations assigned to this cluster are associated with  warmer ambient temperatures and a moderate light absorption coefficient, suggesting that PM2.5 is dominated by organic material rather than secondary inorganic components from the other clusters. Organic-dominated compositions are common in urban environments, where everyday human activity contributes substantially to PM2.5 mass \citep{auckland_comp, egypt}. The dominance of this cluster is consistent with the monitoring network, as most sites are located in cities and towns.

The lower ammonium ($\text{NH}_4^+$) and nitrate $(\text{NO}_3^-)$ proportions in clusters three and four are associated with warmer conditions. Ammonium ($\text{NH}_4^+$) in ambient PM2.5 has been linked in previous studies to agricultural activities such as fertiliser application and livestock waste \citep{ammonium_source,ammonium_nitrate_volatile}. These examples provide context for the types of emissions that may contribute to elevated ammonium ($\text{NH}_4^+$) levels. In contrast, clusters 2 and 4 display higher light absorption coefficients alongside greater elemental carbon compositions, indicating a larger composition of light-absorbing particulate material.

Overall, these clusters distinguish between fresh combustion pollution, chemically aged pollution, and pollution formed under different atmospheric conditions. It is a descriptive view of incomplete proportional data and allows for further long-term investigation. By grouping observations with similar compositions, the clustering approach provides an interpretable summary of heterogeneity in PM2.5 data.

\section{Conclusion} 
A key challenge in compositional data analysis arises when observations are incomplete. While finite mixtures of Dirichlet distributions offer a flexible and interpretable approach for modelling heterogeneous compositional populations, existing likelihood-based approaches do not readily extend to settings with missing or censored components on the unit simplex. As a result, practitioners often discard incomplete observations or apply transformations to unconstrained spaces. These approaches can influence interpretability of compositional data \citep{alr_em_algorithm}. The method accommodates a CAR mechanism, which includes special cases of coarsening, such as coarsening completely at random (CCAR), missing at random (MAR), and censoring. Simulation studies demonstrate that the proposed approach substantially improves model selection and clustering performance relative to maximum likelihood estimation based solely on complete cases. The algorithm more often identifies the correct number of mixture components and produces more accurate cluster assignments than existing approaches. These improvements persist even under severe data incompleteness, including observations with only a single observed component and datasets containing up to 90\% of unobserved values.\\
The proposed methodology was also applied to two real datasets involving compositional variables subject to physical additive constraints. In both cases, the algorithm successfully estimated finite mixture models and identified meaningful clusters directly on the simplex. Because the analysis remains on the simplex, the resulting clusters can be interpreted directly in terms of the original compositional variables, without requiring transformations or ad hoc adjustments.
More broadly, this work reinforces the role of mixtures of the Dirichlet distribution as suitable models for compositional data. When incomplete compositions arise, analysts have historically relied on methods designed for unconstrained data, adapting the data rather than the model. The framework proposed here reverses this perspective by developing estimation tools that respect the geometry and constraints of compositional data. By enabling likelihood-based inference and clustering directly on the simplex in the presence of unobserved parts, this work opens the door to further developments in simplex-based modeling, including richer mixture models and more flexible inference procedures.

\label{conclusion}

\section*{Declarations}
\begin{itemize}
\item Funding:\\
This work is supported in part by the Centre of Excellence in Mathematical and Statistical Sciences, 363 based at the University of the Witwatersrand (SA), grant number PMDS230705128094 as well as the Department of Research and Innovation (DRI).The opinions expressed and conclusions arrived at are those 364 of the authors and are not necessarily to be attributed to the NRF.\\
\\
Antonio Punzo acknowledges the support by the Italian Ministry of University and Research (MUR) under the PRIN 2022 grant number 2022XRHT8R (CUP: E53D23005950006), as part of “The SMILE Project: Statistical Modelling and Inference to Live the Environment”, funded by the European Union – Next Generation EU.
\item Conflict of interest/Competing interests\\
Not applicable
\item Ethics approval and consent to participate\\
Not applicable
\item Consent for publication\\
Not applicable
\item Data availability:\\
The xenolith To-go dataset can be access from the EarthChem repository: \url{https://search.earthchem.org/datatogo}\\
%The American Time Use dataset can be downloaded from the US Bureau of Labor Statistics website, and is under the list of Basic files for the year 2024: \url{https://www.bls.gov/tus/data/datafiles-2024.htm}\\
Samples collected from the air through monitors are tested for numerous molecule species linked to meteorology, pollution, and toxicity. Historic measurements can be accessed through the AQS website:\url{https://www.epa.gov/outdoor-air-quality-data}
\item Materials availability\\
Not applicable
\item Code availability 
Not applicable
\item Author contribution\\
Not applicable
\end{itemize}

\bibliography{main}
\end{document}